\documentclass{article} 
\usepackage{iclr2027_conference,times}
\usepackage{tikz}
\usepackage{tikz-3dplot}
\usepackage{caption}
\usepackage{amssymb}

\usetikzlibrary{arrows.meta,calc,positioning}

\definecolor{cubeFront}{HTML}{BFD3E6}
\definecolor{cubeSide} {HTML}{8FAFCC}
\definecolor{cubeTop}  {HTML}{DCE7F2}
\definecolor{vecPool}  {HTML}{9FC8B8}
\definecolor{vecLogit} {HTML}{E8B98A}
\definecolor{edgeCol}  {HTML}{3E4A56}
\definecolor{txtCol}   {HTML}{2B2F33}

\usepackage{amsmath,amsfonts,bm}

\def\eqref#1{equation~\ref{#1}}
\def\1{\bm{1}}

\DeclareMathAlphabet{\mathsfit}{\encodingdefault}{\sfdefault}{m}{sl}
\SetMathAlphabet{\mathsfit}{bold}{\encodingdefault}{\sfdefault}{bx}{n}

\newcommand{\E}{\mathbb{E}}

\DeclareMathOperator{\sign}{sign}

\usepackage{hyperref}
\usepackage{url}
\usepackage{booktabs} 
\usepackage{multirow}
\usepackage{amsthm}
\usepackage{pifont}
\newcommand{\cmark}{\ding{51}}
\newcommand{\xmark}{\ding{55}}

\numberwithin{appendixprop}{subsection}
\newtheorem{appendixdef}{Definition}
\numberwithin{appendixdef}{section}

\title{Optimizing and Securing the Modern Watermarking Channel for Images}

\author{Enoal Gesny \& Eva Giboulot \\
Inria, Rennes, France\\
\texttt{\{enoal.gesny\}@inria.fr} \\
}

\iclrfinalcopy 
\begin{document}

\newcommand{\evag}[1]{\textcolor{purple}{#1}}
\newcommand{\enoal}[1]{\textcolor{blue}{#1}}
\newcommand\diffmod[1]{\diffmodOP\left({#1}\right)}
\def \Hzero{\mathcal{H}_0}
\def \Hone{\mathcal{H}_1}
\def\un{\mathbbm{1}}
\def\VAE{\mathrm{VAE}}
\def\ie{\textit{i.e.}}
\def\sign{\mathrm{sign}}
\def\PFA{\mathrm{P}_{\mathrm{FA}}}
\def\PD{\mathrm{P}_{\mathrm{D}}}

\newcommand{\method}{\textcolor{blue}{METHOD }}

\def\VSS{\texttt{V-SSig}}
\def\VVS{\texttt{V-VS}}
\def\VTM{\texttt{V-TM}}

\def\GSS{\texttt{G-SSig}}
\def\GVS{\texttt{G-VS}}
\def\GTM{\texttt{G-TM}}

\def\GVSS{\texttt{GV-SSig}}
\def\GVVS{\texttt{GV-VS}}
\def\GVTM{\texttt{GV-TM}}

\def\SS{\texttt{SSig}}
\def\VS{\texttt{VS}}
\def\TM{\texttt{TM}}

\def\method{\textsc{SNW}~}

\def\bx{\mathbf{x}}
\def\bz{\mathbf{z}}
\def\bk{\mathbf{k}}
\def\bv{\mathbf{v}}
\def\bb{\mathbf{b}}

\def\bU{\mathbf{U}}

\newtheorem{proposition}{Proposition}
\newtheorem{property}{Property}
\newtheorem{corollary}{Corollary}
\newtheorem{definition}{Definition}

\def\keyset{\mathcal{K}}
\def\E{\mathbb{E}}
\def\P{\mathbb{P}}
\def\V{\mathbb{V}}

\def\partition{\mathcal{P}}

\def\Pe{\mathrm{P}_e}

\def\bx{\mathbf{x}}
\def\bX{\mathbf{X}}
\def\bY{\mathbf{Y}}
\def\bZ{\mathbf{Z}}
\def\bW{\mathbf{W}}
\def\bI{\mathbf{I}}
\def\Mprime{M^{\prime}}

\def\bxatk{\mathbf{x}_{\mathrm{atk}}}
\def\bxwm{\mathbf{x}_{\mathrm{wm}}}
\def\btxwm{\tilde{\mathbf{x}}_{\mathrm{wm}}}

\def\bm{\mathbf{m}}
\def\bc{\mathbf{c}}
\def\bu{\mathbf{u}}
\def\bU{\mathbf{U}}
\def\bn{\mathbf{n}}
\def\bw{\mathbf{w}}
\def\bq{\mathbf{q}}

\def\DKL{\mathrm{D}_{\mathrm{KL}}}

\def\beps{\boldsymbol{\epsilon}}
\def\bmu{\boldsymbol{\mu}}
\def\bsigma{\boldsymbol{\sigma}}
\def\bSigma{\boldsymbol{\Sigma}}

\def\DKL{\mathrm{D}_{\mathrm{KL}}}

\usetikzlibrary{arrows.meta, calc, fit, positioning, backgrounds}
 
\definecolor{blueBox}{RGB}{210,228,252}
\definecolor{blueBorder}{RGB}{70,130,210}
\definecolor{blueTitle}{RGB}{50,110,190}
 
\definecolor{redBox}{RGB}{252,215,215}
\definecolor{redBorder}{RGB}{210,80,80}
\definecolor{redTitle}{RGB}{190,50,50}
 
\definecolor{greenBox}{RGB}{215,240,220}
\definecolor{greenBorder}{RGB}{70,165,80}
\definecolor{greenTitle}{RGB}{50,140,60}
 
\definecolor{arrowGray}{RGB}{60,60,60}
 
\newlength{\imgW}\setlength{\imgW}{1.6cm}
\newlength{\imgH}\setlength{\imgH}{1.6cm}
 
\tikzset{
  myarrow/.style={->, >=Stealth, thick, dashed, arrowGray,
                  shorten >=2pt, shorten <=2pt},
  solidarrow/.style={->, >=Stealth, thick, arrowGray,
                     shorten >=2pt, shorten <=2pt}
}

\maketitle

\begin{abstract}
To comply with recent regulations requiring traceable generated content, modern watermarking has adopted multi-bit post-hoc watermarking schemes. These modern designs rest on an encoder-decoder pair implemented as deep neural networks.
These models are usually treated as pure black-boxes trained end-to-end, with the noise of the watermarking channel modeled through a fixed set of geometric and valuemetric transforms applied to watermarked images. 
We argue that this purely empirical approach leads to unquestioned design flaws and a lack of theoretical performance guarantees.
This work proposes a general theoretical model of modern post-hoc watermarking schemes grounded in a statistical analysis of the outputs of the encoder/decoder pair. We show that these deep neural networks implicitly define a watermarking channel modeled as parallel AWGN channels, with messages transmitted using BPSK modulation. This imposes a binary alphabet, greatly limiting the capacity of these watermarking systems. 
Another fatal flaw is their lack of a secret key, making them intrinsically insecure. We make this notion of watermarking security precise for post-hoc schemes by linking it to the possibility of estimating the secret key under a given statistical model of the decoder's output.
By putting together the results from this theoretical analysis, we introduce SNW: a novel post-hoc watermarking system that significantly outperforms existing state-of-the-art baselines in terms of capacity while also providing strong security guarantees. Notably, it does not depend on a fixed codebook or binary alphabet, allowing it to reach a rate close to Shannon capacity through the use of capacity-achieving error-correcting codes.
\end{abstract}
\section{Introduction}

Digital media watermarking is going through a striking revival since the early 2020's, accompanying the usage growth of generative models. At the time of writing, the EU AI Act~\cite{european_data_protection_supervisor_ai_2025} has entered into full force, along with its Code of Practice~\cite{ai_office_general-purpose_2025}, making the use of marking technology mandatory for providers of generative models and services.

For this task, post-hoc watermarking has been largely favored by the industry over in-generation techniques. An instructive evidence is the SynthID report for image watermarking~\cite{gowal_synthid-image_2025}. It explicitly defines what is expected, in Google DeepMind's view, from a watermarking system -- Section 2.1 -- and singles out multi-bit  post-hoc watermarking as the most practical approach -- Section 2.2. 

This focus of the industry on multi-bit post schemes is also reflected in the open-source literature: Meta's \textsc{Seal} family~\cite{fernandez_video_2024, petrov_we_2025,soucek_pixel_2025} and Adobe's \textsc{Trustmark}~\cite{bui_trustmark_2025} proposed \textit{almost only} post-hoc multi-bit schemes\footnote{The only exception being \textsc{DistSeal}~\cite{rebuffi_learning_2026}, a seed-based in-generation scheme.}.

Such a convergence from the industry warrants a rigorous and critical evaluation methodology for this family of watermarking systems, something which, we believe, is currently missing in the literature.

In order to introduce our argument, it is fruitful to distinguish between two categories of post-hoc schemes: classical and modern. The former, exemplified by Spread-Spectrum~\cite{cox_secure_1997} and Broken-Arrows~\cite{furon_broken_2008}, rely on handcrafted, often invertible and orthonormal transforms to construct their watermarking space. The modern approach was pioneered in 2018 by the \textsc{HiDDeN} architecture and training recipe~\cite{ferrari_hidden_2018}, and refined across the 2020's with the aforementioned \textsc{Seal} family, \textsc{Trustmark} and \textsc{SynthID-Image}. The main idea rests on training end-to-end a pair of encoder/decoder deep neural networks (DNN). Noise in the watermarking channel is modelled with a finite set of image augmentations, applied after encoding. The pair of DNN is optimized to be as robust as possible to these augmentations, while keeping the watermarking signal imperceptible.

There is a longstanding problem in classical watermarking: one cannot be robust \textit{at the same time} to all geometric transforms. For example, if one builds a watermarking space invariant to rotation, scaling, and translation using the Fourier-Mellin transform, one loses robustness to cropping\cite{bas_geometrically_2002}. Solving this so-called geometric robustness problem has often relied on synchronization strategies which bring their own difficulties.

The promise of the modern approach is to solve the geometric robustness problem, \textit{for free}. By simply adding both scaling and cropping to the set of augmentations during training, \textsc{VideoSeal} empirically demonstrated~\cite{fernandez_video_2024}[Appendix B.2] that robustness to both operations was indeed possible and was observed, to some extent, for every modern approach.

Despite this breakthrough, a number of recent works have been pointing out major security flaws in the design of modern watermarking schemes, putting into question the "free-lunch" claim of these schemes. These critiques relate to two main weaknesses: 1) the lack of a secret key to hide the message, making the watermarking system highly vulnerable under Kerckhoffs' principle~\cite{bas_ai_2025, tarhini_neural_2026} and 2) the reliance on DNN which largely opens up the attack surface related to adversarial security~\cite{jiang_evading_2023, kassis_unmarker_2025, gesny_modern_2026}. These works demonstrated that, for these systems, it is not only possible but simple and inexpensive to erase, copy, or forge a watermarking signal. 

This also questions the value of modern schemes beyond the geometric problem: classical schemes seem to perform as well in terms of robustness against valuemetric augmentations (JPEG compression, linear filtering, saturation, $\ldots$), yet they are far less vulnerable to both white-box and black-box attacks~\cite{gesny_modern_2026}. The problem seems even more pressing given that, of all the modern schemes cited so far, only one, \textsc{SynthID-Image}, \textit{mentions} the notion of security.

Another limitation of the modern literature is methodological. Modern systems all rely on the same end-to-end training recipe, with innovations being mostly architectural. A case in point is ChunkySeal~\cite{petrov_we_2025}: its authors are surprised that the theoretical watermarking capacity achievable for images is far from being attained by current schemes. Their answer is to train a larger model. In Section~\ref{sec:channel} we demonstrate that the inefficiencies of current schemes actually stem from implicit design choices induced by the architecture and training pipeline. We then show how to \textbf{increase the capacity of modern schemes almost threefold} without increasing the size of the encoder/decoder, nor changing their architecture, and with significantly less training compute.

The main goal of this work is to clarify and systematize the design of post-hoc watermarking systems. Notably, we want to emphasize the security aspect of watermarking, which is too often neglected in the modern literature: 

\begin{itemize}
\item We propose a general framework for defining and evaluating each component of modern post-hoc systems -- Section~\ref{sec:formalism}. In particular, we propose a formal distinction between watermarking security and adversarial robustness.
\item We provide a theoretical model of the modern post-hoc watermarking channel as an Additive White Gaussian Channel (AWGN), with the watermarking signal transmitted using Binary Phase-Shift Keying (BPSK). We empirically validate this model on state-of-the-art watermarking systems, demonstrating at the same time fundamental limitations of this design -- Section~\ref{sec:channel}.
\item We design a novel watermarking system, Secure Neural Watermarking (SNW), which integrates a secret key in the embedding and decoding mechanisms. It reaches a maximum capacity of 656 bits with a binary alphabet, significantly outperforming the 256 bits current state-of-the-art with the same DNN architecture, while being perfectly secure against PCA attacks -- Section~\ref{sec:method}-\ref{sec:snw_theoretical}.
\end{itemize}

\section{Multi-bit Watermarking systems}\label{sec:formalism}

\subsection{Problem formulation}
 Cox defines multi-bit watermarking as the "reliable transmission of a message over an unreliable channel"~\cite{cox_watermarking_2006}.
 A useful illustrative scenario for this problem is the attribution scenario that we adapt from~\cite{gesny_secure_2026}[Section 2]:

 \paragraph{Watermarking attribution scenario} Alice supplies an API where users can request images to be generated. In order to trace the use of her system, she asks Bob, a third-party, to provide her with a secret key $k$ and a post-hoc watermarking system $\mathcal{W}$. From the point of view of Bob, the secret key $k$ is now uniquely linked to Alice. She then associates a unique ID $\bm_{\mathrm{user}}$ with each user. Each request generates an unwatermarked image which is passed to $\mathcal{W}$ in order to watermark it. The client is served only the watermarked content. If Bob is presented with the key $k$ and a watermarked image generated for user $i$, he should decode the correct message $\bm_i$. If the image was not watermarked, or the key is incorrect, the decoded message should be random.%\footnote{See Def~\ref and its scholia.}.

Our threat model is focused on spoofing attacks under Kerchoff's principle. Note that shifting to a threat model on erasure attacks requires only slight changes in the analysis, which mostly amounts to working on watermarked images instead of non-watermarked ones. 

\paragraph{Threat model -- Spoofing Attack} Eve is a malicious third party who observes a collection of $N$ pristine watermarked images created by independent users of Alice's API. 
Eve aims to spoof a specific user's identity: Camille's. Her identity is recorded as the message $\bm_{\text{Camille}}$ through the codeword $\bc_{Camille}$. 
Kerckhoff, a malicious colleague of Alice, provides Eve with complete access to the watermarking pipeline \textbf{except} Alice's secret key $k$. We further assume that Eve has access to at least one image generated by Camille -- though she does not know $\bc_{Camille}$ a priori. We assume that when decoding an image, Eve does not introduce any noise (i.e., perfect channel assumption).

\subsection{Main definitions}

\paragraph{Message vs Codeword}For a $M$-bit multi-bit system, it is important to distinguish a message $\bm \in \{0,1\}^{M}$ from a codeword $\bc$. A codeword is the object that is eventually mapped into a message using an error-correcting code -- see the redundancy mechanism definition Def~\ref{def:redundancy_mechanism}. In our case, a codeword can be either a sequence of bits in $\{0,1\}^{M^{\prime}}$ or a real vector in $\mathbb{R}^{M^{\prime}}$. We impose that $M^{\prime} \geq M$. We sometimes leave the alphabet of the codeword unspecified, in which case we denote it as $\mathcal{A}^{M^\prime}$.

\paragraph{Post-hoc watermarking systems} We adapt the general formulation from~\cite{gesny_secure_2026}[Def. B.2] to post-hoc watermarking systems $\mathcal{W}$ as:
 \begin{itemize}
 	\item A set of secret keys $\mathcal{K}$
 	\item A family of \textit{deterministic} embedding functions $(e_k)_{k\in \mathcal{K}}$ : $\mathbb{R}^{L_e} \times \mathcal{A}^{M^\prime} \rightarrow \mathbb{R}^L$, embedding the watermark signal into the host content, with respect to the secret key $k$. 
 	\item Two forward projection functions $f_e: \mathbb{R}^{D} \to \mathbb{R}^{L_e}$ and $f_d: \mathbb{R}^{D} \to \mathbb{R}^L$, projecting the content from pixel space into watermarking space. 
 	\item A backward projection function $f_e^{\dagger}: \mathbb{R}^{L_e} \to \mathbb{R}^{D}$
 	\item A family of decision mechanisms $(d_k)_{k\in \mathcal{K}}: \mathbb{R}^L \to \mathcal{A}^{M^\prime}$, extracting the codeword $\bc$ from an observation in watermarking space using the secret key $k$. 
 \end{itemize}

For the rest of the paper, we denote by $\mathcal{F}$ the distribution of cover images. Following our watermarking scenario, the cover distribution should be such that it is not biased towards a given message -- see Def.~\ref{def:cover_distrib}.

\subsection{Watermarking security vs Adversarial robustness}\label{subsec:wmsec}
A post-hoc watermarking system offers two main attack surfaces for Eve in our spoofing scenario: 1) unauthorized access to the watermarking channel if the secret key is stolen and 2) copying the latent vector $f_d(\bxwm)$ of an image for which $\bm$ is known using an adversarial example.

The current literature usually conflates these two approaches into a single concept of security. We argue in Appendix~\ref{app:secu_adv_robustness} that they are two very different problems, with different consequences. We thus distinguish between true watermarking security, pertaining to embedding/decoding mechanism $(e_k, d)$, and adversarial robustness, which pertains to the decoding projection function $f_d$.

\begin{definition}[Watermarking Security]\label{def:security}
    Let $\mathcal{W}$ be a watermarking system.
    Let $X^{(N)}$ be a set of $N$ independent images watermarked with a system $\mathcal{W}$. Denote $Z_i := f_d(X_i)$.
    
    Let $\psi_N : (\mathbb{R}^L)^N \rightarrow \mathcal{K}$ be a function that estimates a secret key from a set of watermarked observations. 
    
    For a given codeword $\bc$ and unwatermarked image $\bx$, denote the probability of attack success as:
    \begin{equation}
        \mathrm{P}_{\mathrm{atk}}(N, \bx) := \P [d(f_d(\bxatk^{(N)}, k)) = \bc],
    \end{equation}
    where $\bxatk^{(N)} := f_e^{\dagger}(e_{\psi_N(Z^{(N)})}(f_e(\bx), \bc))$.
    The watermarking system $\mathcal{W}$ is said to be $\eta L$-secure iff:
    \begin{equation}
        \E\left[P_{\mathrm{atk}}(\eta L,Y)\right] \leq \E\left[\P[d(f_d(Y),K) = C]\right],
    \end{equation}
    where the expectation is taken over $(Y, C, K)$ triplets, with $\bx$ sampled from the cover distribution $\mathcal{F}$, and $(C, K)$ sampled uniformly from their respective set.
\end{definition}

In words, a watermarking system is $\frac{N}{L}$-secure if, on average, providing Eve with $N$ independent observations does not improve her spoofing success beyond random chance. We will only focus on the watermarking security aspect; adversarial robustness is a wholly different problem that needs to be tackled with adversarial machine learning tools, which is out of the scope of this paper (but see Appendix~\ref{app:secu_adv_robustness} for an evaluation of the current state of things).

\section{The modern post-hoc watermarking channel}\label{sec:channel}
\begin{figure}[ht]
    \centering
    \includegraphics[width=0.5\linewidth]{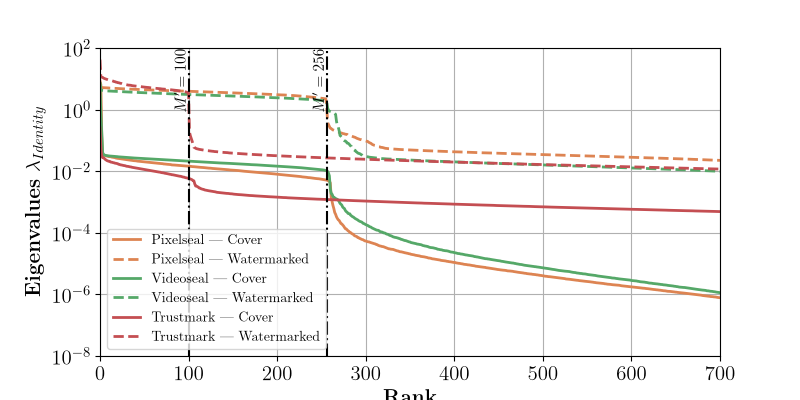}%
    \includegraphics[width=0.4\linewidth]{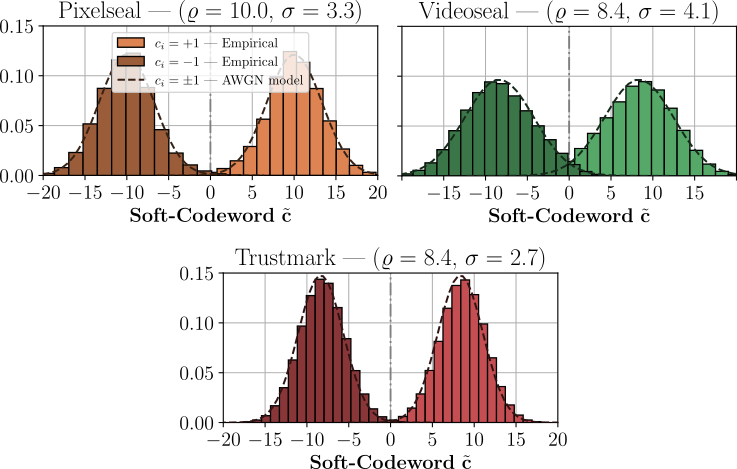}
    \caption{Empirical analysis of the decoding projection $f_d$ and decision mechanism $d$ of three SOTA post-hoc watermarking systems computed over 10k $1024 \times 1024$ ImageNet~\cite{deng2009imagenet} images watermarked with random messages.\textit{ (Left)} Eigenvalues 
    of the covariance $\Sigma_{\mathrm{Identity}}$ of the latent vector $f_d(\bxwm)$ (i.e. when $t$ is the identity). \textit{(Right)} Global distribution of the soft-codeword values $\tilde{\bc}$ and the corresponding theoretical AWGN model with $(\varrho,  \sigma)$ computed empirically.}
    \label{fig:modern-channel-study}
\end{figure}

\begin{table}[ht]
    \centering
    \small 
    \resizebox{1\linewidth}{!}{
    \begin{tabular}{l ccccc | ccccc}
        \toprule
        & \multicolumn{5}{c}{\textbf{Channel characteristic} $\varrho\sigma^{-1}$($\uparrow$) }  &  \multicolumn{5}{c}{\textbf{Capacity} $C_(\varrho,\sigma)$ / \textbf{Empirical Capacity}}\\
        \cmidrule(lr){2-11}
         & Identity & JPEG 50 & Contrast $\times 2$ & Crop 0.6 & Rot. 90$^\circ$ & Identity & JPEG 50 & Contrast $\times 2$ & Crop 0.6 & Rot. 90$^\circ$ \\
        \midrule
        \textsc{PixelSeal}     & 3.07& 2.21 & 1.60  & 1.73 & 1.97 & 0.99 / 0.96 & 0.90 / 0.88 & 0.69 / 0.72  & 0.75 / 0.75 & 0.84 / 0.79 \\
        \textsc{VideoSeal}     & 2.03 & 1.64 & 1.03 & 0.71 & 1.17 & 0.85 / 0.85 & 0.71 / 0.75 & 0.38 / 0.47 & 0.21 / 0.21 & 0.74 / 0.75 \\
        \textsc{TrustMark}     & 3.11 & 2.67 & 1.26 & 0.26 & 0.20 & 0.99 / 0.98 & 0.96 / 0.95 & 0.52 / 0.58 & 0.03 / 0.00 & 0.02 / 0.00 \\
        \bottomrule
    \end{tabular}
    }
    \caption{Channel characteristic $\varrho\sigma^{-1}$ computed empirically over 10k $1024 \times 1024$ ImageNet images watermarked with random messages and the corresponding theoretical capacity in message bits per codeword bit.}
    \label{tab:channel-charac}
\end{table}
Let $\bx$ be an unwatermarked image and $t : \mathbb{R}^D \rightarrow \mathbb{R}^D$ a transform in pixel space (which can be the identity function).
The deep neural networks serving as the projection functions of modern watermarking systems all rely on the same design:

\paragraph{Embedder $(f_e, f_e^\dagger, e)$} The forward and backward embedding projections are implemented with a single auto-encoder\footnote{The HiDDeN embedder was more complicated, but recent designs have all streamlined the embedding down to one single auto-encoder.}. It is composed of a downsampling stage $f_e$ -- the forward embedding projection  -- a bottleneck, where current systems usually apply the embedding function $e$, and an upsampling stage $f_e^\dagger$ -- the backward embedding projection. The embedding function $e$ maps the codeword $\bc \in \{0,1\}^{\Mprime}$ to what we call a \textit{steering vector} $\bv(\bc) : \{0,1\}^{\Mprime} \rightarrow \mathbb{R}^{\Mprime}$. This mapping is fixed in advance, for example, VideoSeal and PixelSeal use WAM's method~\cite{sander2025watermark} of mapping the $i$-th component $c_i$ to one of two a priori fixed number $(v^-_i, v^+_i)$ depending on the parity of $c_i$. The steering vector $\bv(\bc)$ is then mixed within the downsampling features $f_e(\bx)$. To summarize, the embedding pipeline can be written as:
\begin{align}\label{eq:modern-steering}
    \mathbf{w} &= f_e^{\dagger}(e(f_e(\bx), \bc)) \text{ where } e(\cdot, \bc) := \mathrm{Mix}(\cdot, \bv(\bc))\\
    \bxwm &= \bx + 10^{\frac{\gamma}{20}} \sqrt{\mathbf{w}}||\mathbf{w}||_2^{-1},
\end{align}
where $\gamma$ controls the PSNR of the watermarking signal and $\mathrm{Mix}$ is an unspecified function mixing the steering vector with the input (in practice implemented as concatenation). Before reaching the decoder, the watermarked image gets perturbed by the transform $t$: $\btxwm := t(\bxwm)$.

\paragraph{Decoder $f_d$} The decoding stage is always implemented as a convolutional neural network (CNN) acting as a bit classifier.
The decoding projection itself encompasses the main CNN module and the final pooling operation that outputs a $L$-dimensional latent vector $\bz =f_d(\btxwm)$. Empirically, we observe that $\bz$ can be modeled as following a multivariate Gaussian distribution $\mathcal{N}(\bmu_z, \Sigma_t)$. In practice, none of the studied systems have $L = \Mprime$. Consequently, we must assume $\Sigma_t$ to be only semi-definite, i.e the Gaussian may be degenerate.

\paragraph{Decision mechanism $d$} Finally, a linear head further reduces the latent vector down to what we call the soft-codeword $\tilde{\bc}$.
The final codeword $\bc$ is then decoded by applying the sign function: $\bc = \mathrm{sign}(\bW\bz + \bb)$.

In order to understand the role of the linear head, we turn our attention to the structure of $\Sigma_t$ in Figure~\ref{fig:modern-channel-study} (left). As expected, the rank of these covariance matrices is never $L$. One can observe a "step-like" behavior within the eigenvalues: $M^{\prime}$ of them are high, while the rest quickly goes to zero. We call the former the \textit{robust components} of the latent vector. Though it could be mistaken for estimation noise, there is indeed a small subset of non-zero eigenvalues with smaller magnitude, these we call the \textit{brittle components}. The rest of the eigenvectors have eigenvalues of zero: these dimensions are purely \textit{redundant}.

The fact that $\Sigma_t$ is rank deficient sheds some light on role of the linear head: it performs a whitening operation, getting rid of the redundant and brittle components, as well as unbiasing the codeword estimation. We propose to model $\tilde{\bc}$ as following a multivariate Gaussian with diagonal covariance $\mathcal{N}(\varrho_t \bc, \bsigma_t\bI_{\Mprime})$ -- see Figure~\ref{fig:modern-channel-study} (right).

From this analysis, we claim that the watermarking channel induced by the encoder/decoder DNN pair can be modeled with $\Mprime$ parallel AWGN channels and a watermarking signal modulated using BPSK. We record here the well-known results about such communication systems:

\begin{proposition}[Capacity of AWGN channel with BPSK]
    The capacity $C_{\varrho, \sigma}$ of an AWGN channel with variance $\sigma^2$  and signal transmitted with power $\varrho$ and BPSK modulation is
        $C_{\varrho, \sigma} = h_2\left(1-\Phi\left(-|\varrho|\sigma^{-1}\right)\right)$,
    where $h_2$ is the binary entropy function, and $\Phi$ the standard Gaussian c.d.f.
    Notably, $\forall \varrho > 0, \forall \sigma \geq 0, C_{\varrho, \sigma} \leq 1 $.
\end{proposition}

In the case of a binary alphabet, BPSK modulation is optimal for an AWGN channel in the sense that it maximizes the theoretically achievable capacity; but it limits the capacity to 1 bit/channel element. We report the $(\varrho, \sigma)$ for three state-of-the-art watermarking systems and validate the predictivity of this model in Table~\ref{tab:channel-charac}.

\paragraph{Security of keyless systems} Modern watermarking systems following this model have no security since they have no secret key. Even if we allow their linear head to be kept secret -- disregarding Kerchoff's principle -- they are still vulnerable to a linear estimation attack: Eve can simply watermark $L+1$ images with $L+1$ independent messages (the encoding does not depend on a secret key, nor on the linear head). She then applies the decoding projection to these images, giving her $L+1$ latent vectors $\bz$. Since she knows $\bc$ and $\bz$, she can solve the well-defined system $\bW\bz + \bb = \bc$ for $\bW$ and $\bb$. She can then recover Camille's codeword and spoof as many images as she wishes.
\section{\method: Secure Neural Watermarking}
\label{sec:method}

So far, we have outlined three main limitations of modern watermarking systems: 1) lack of a secret key, making spoofing attacks trivial, 2) inefficient use of the watermark space, with redundant latent components, and 3) a capacity limited to 1 message bit per codeword bit. To address these, we propose \method, a secure post-hoc watermarking design that allows the use of a continuous alphabet and is designed to maximize the use of watermark space.

\subsection{Secure decision mechanism and embedding functions $(e_k, d)$}
In order for our system to be secure, we integrate a secret key into both the embedding and decision mechanisms. We define the secret key set $\keyset$ as the set of all semi-orthogonal matrices of dimension $L \times \Mprime$, denoted $\bU$, i.e. $ \forall \bU \in \keyset, \bU^T\bU = \bI_{\Mprime}$. Let $\bc$ be a codeword from a \textit{continous} alphabet $\mathbb{R}^{\Mprime}$. Analogously to Eq.(\ref{eq:modern-steering}), we define the embedding function $e_\bU$ using a normalized \textit{steering vector} $\bv \in \mathbb{S}^{L-1}$ constructed by distributing the watermarking energy between the secret subspace defined by $\bU$ and its orthogonal complement:

\begin{align}
    e_\bU(\bz, \bc) = \mathrm{Mix}(\bz, \bv) \quad ; \quad
     \bv := \alpha \frac{\mathbf{U}\bc}{\sqrt{M^\prime}} + \sqrt{1 - \alpha^2} \frac{(\mathbf{I}_L - \mathbf{U}\mathbf{U}^\top)\bq}{\|(\mathbf{I}_L - \mathbf{U}\mathbf{U}^\top)\bq\|_2},
     \label{eq:steering_vector}
\end{align}

Where $\bq$ is realization of a standard Gaussian random variable $\mathcal{N}(0, I_L)$ and $\alpha \in [0,1]$. We show in Section~\ref{sec:snw_theoretical} that $\alpha$ controls the trade-off between capacity and security.

Assuming $f_d$ can perfectly retrieve $\bv$, the decision mechanism $d_\bU$  extracts $\bc$ simply by applying the secret rotation $\bU$ to $\bv$: $d_\bU = \bU^T\bv = \bc \alpha\sqrt{\Mprime}^{1}$. 
Note that, by further applying the $\mathrm{\sign}$ function after $d_\bU$, we retrieve the BPSK modulation in Section~\ref{sec:channel}. Yet, since $\bc$ is real-valued, we are also free to use any off-the-shelf capacity-achieving codes for AWGN channels such as nested lattice codes~\cite{zamir_lattice_2014}.

\subsection{Learning isotropic and robust neural projections $(f_e, f_e^\dagger, f_d)$}

Our decoding projection needs to retrieve the steering vector $\bv$ with the least amount of noise possible. However, contrary to other modern designs, we allow embedding any arbitrary vector on the $L-1$ hypersphere. Furthermore, in order to maximize capacity, we don't want the decoding function to favor certain regions of the hypersphere, i.e. we want the distribution of $f_d$'s output to have a covariance $\Sigma$ with full rank.

To achieve this,  we optimize jointly $(f_e, f_e^{\dagger}, f_d)$ through the following objective balancing alignment, latent space isotropy, and image quality:

\begin{equation}
    \mathcal{L} = \lambda_{\text{align}} \mathcal{L}_{\text{align}}(f_d(\bxwm), \bv) + \lambda_{\text{iso}}\mathcal{L}_{\text{iso}}(f_d(\bxwm), f_d(\bx)) + \lambda_{\text{qual}} \mathcal{L}_{\text{qual}}(\bx, \bxwm).
    \label{eq:loss}
\end{equation}

 \paragraph{Alignment $\mathcal{L}_{\text{align}}$} Maximizes the cosine similarity between the extracted representation $\bz$ and the target vector $\bv$. 
    \paragraph{Isotropy $\mathcal{L}_{\text{iso}}$} Minimizes correlations by penalizing non-zero inner products across non-watermarked images or images watermarked with different codewords. This enforces a uniformly distributed spherical latent space. 
    \paragraph{Quality $\mathcal{L}_{\text{qual}}$} Preserves the visual fidelity of the host content by combining a constraint on the PSNR of the watermark power and an LPIPS perceptual metric~\cite{zhang2018perceptual}. 

We provide complete definitions of all loss components in Appendix~\ref{app:losses} and details on the training procedure in Appendix~\ref{app:training_details}.
\section{\method theoretical analysis}\label{sec:snw_theoretical}

\subsection{Capacity}
Let $\bxwm$ be an image watermarked with \method and $t: \mathbb{R}^D \rightarrow \mathbb{R}^D$ a transform in pixel space (which can be the identity). Denote $\tilde{\bv} := (f_d \circ t) (\bxwm)$ the normalized steering vector extracted from the transformed $\bxwm$.
We model $\tilde{\bv}$ as a perturbation of the true steering vector $\bv$:
\begin{equation}
    \tilde{\bv}= \rho \bv + \sqrt{1 - \rho^2} \, \mathbf{n}, \quad \text{with } \mathbf{n} \sim \mathcal{U}(\mathbb{S}^{L-2}), \quad \mathbf{n} \perp \bv.
    \label{eq:perturbation_model}
\end{equation}

In practice, note that due to the decoding projection $f_d$ imperfections, even if $t$ is the identity function, we don't expect $\rho=1$.

We now provide the \method capacity when using a binary alphabet for the codeword $\bc$:

\begin{proposition}[Binary \method capacity]
    Under the perturbation model defined in Equation~\ref{eq:perturbation_model}, the probability $p(\rho)$ of correctly decoding a bit is given by:
    \begin{equation}
        p(\rho) = \Phi\left(\alpha\frac{\rho}{\sqrt{1 - \rho^2}}\sqrt{\frac{L}{M^\prime}} \right),
    \end{equation}
    \label{prop:bit_acc}
    where $\Phi$ is the standard Gaussian c.d.f.
    The resulting Shannon capacity is: $C_\rho = M^\prime \left(1 - h_2(p(\rho))\right)$.
\end{proposition}

\begin{proof}
See Appendix~\ref{subsec:proof_bit_acc}. 
\end{proof}

As shown in Figure~\ref{fig:bit_acc_th}, maximizing decoding robustness corresponds to $\alpha = 1$, where all embedding energy is allocated to the watermark signal.

\begin{figure}
    \centering
    \includegraphics[width=\linewidth]{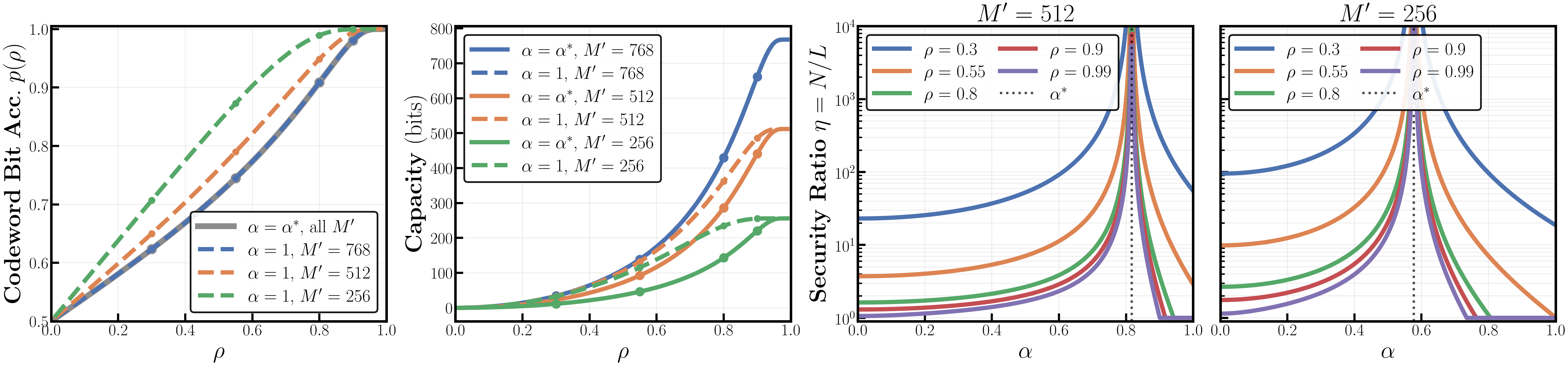}
    \caption{Theoretical bit accuracy $p(\rho)$ (left) and Shannon capacity in bits (right) as a function of the projection cosine alignment $\rho$. Setting $\alpha = \sqrt{M^\prime/L}$ guarantees perfect security against PCA attacks while maintaining an invariant bit accuracy across all codeword dimensions $M^\prime$. Dotted lines represent the maximum robustness setting, when $\alpha = 1$. Security ratio $\eta$ required to estimate the secret subspace with a PCA attack as a function of $\alpha$ for different values of alignment $\rho$ and $\Mprime \in  \{256, 512\}$.}
    \label{fig:bit_acc_th}
\end{figure}

\subsection{Watermarking Security}

The attacker aims to retrieve the secret matrix $\bU$ from $N \geq L$ watermarked images with different codewords. Under our statistical model, Principal Component Analysis (PCA) can be used to do so.
Eve estimates the covariance matrix $\mathbf{\Sigma}_N = \frac{1}{N}\sum_{i=1}^N \bz_i \bz_i^\top$. The expected eigenvalues have two possible values, each spanning subspaces of dimensions $M^\prime$ and $L - M^\prime$, respectively:
\begin{equation}
    \lambda_1 = \rho^2 \frac{\alpha^2}{M^\prime} + \frac{1 - \rho^2}{L}, \qquad
    \lambda_2 = \rho^2 \frac{1 - \alpha^2}{L - M^\prime} + \frac{1 - \rho^2}{L},
    \label{eq:eigenvalues}
\end{equation}

For un-watermarked images, the Marchenko-Pastur distribution gives the asymptotic support $S_0$ of empirical eigenvalues as $N \rightarrow \infty$~\cite{vallet_improved_2012,furon_fast_2013}. 
Furthermore, the theory shows that one cannot distinguish between covariances if the support $S_1$ and $S_2$ of the Marchenko-Pastur distribution associated with $\lambda_1$ and $\lambda_2$ are such that $S_1 \cup S_2 \subseteq S_0$: the PCA is not able to estimate any dimension of the secret key.

\begin{proposition}[\method PCA security]
    A \method system is $\eta L$-secure against PCA attack, requiring at least $N = \eta L$ watermark observations to retrieve the watermark subspace, with:
\begin{equation}
    \frac{N}{L} = \eta = \max\left(1 , \frac{\left( 1 - \sqrt{\lambda_1 M'} \right)^2}{L \left( \sqrt{\lambda_1} - \frac{1}{\sqrt{L}} \right)^2}\right).
\end{equation}
In particular, by setting $\alpha$ to $\alpha^* := \sqrt{\frac{M^\prime}{L}}$ we have that $N \rightarrow \infty$ and the system is said to be \textit{perfectly secure} against PCA attacks.
\label{prop:security}
\end{proposition}
\begin{proof}
    See Appendix~\ref{subsec:proof_security}-\ref{subsec:proof_perfect_security}.
\end{proof}

Note that \method is always perfectly secure when $\alpha=1$ and $L=\Mprime$. This is also the regime that achieves maximum capacity. In other words, when the watermark space is used efficiently -- the codeword size matches watermarking space dimensions -- one should allocate all the energy to the watermarking signal. However, if the codeword size is smaller than $L$, such as in the case of other modern systems, one must waste energy in order to "drown" the watermarking signal into random noise. We validate the complete theoretical analysis empirically in Appendix~\ref{app:SNW-geometric}.
\section{Experiments}

\subsection{Experimental Settings}\label{sec:experimental_settings}

We evaluate the proposed \method watermarking scheme against state-of-the-art recent post-hoc baselines across the capacity-quality-security trade-off. 

\paragraph{Baselines watermarking systems}
We benchmark \method against three recent neural  network watermarking schemes: VideoSeal~\cite{fernandez_video_2024} and PixelSeal~\cite{soucek_pixel_2025}, which embed a 256-bit binary codeword, and TrustMark~\cite{bui_trustmark_2025}, which embeds a 100-bit binary codeword. 
To ensure a fair message-agnostic comparison, we whiten the output representation of all models on a set of $10^{6}$ MFlickr images, following the procedure described in Appendix~\ref{app:whitening}.

\paragraph{\method setup}
We train the \method encoder/decoder pair on the COCO 2017 train dataset~\cite{lin2014microsoft} using the multi-stage pipeline detailed in Appendix~\ref{app:training_details}. 
To isolate the contribution of our watermarking design from architectural choice, \method adopts the same architecture and augmentations as PixelSeal, \textit{except for the linear head}, which is removed.  
Although our framework supports continuous alphabets to achieve higher transmission rates, for fairness we constrain our evaluation to a binary alphabet when comparing with baselines.
Throughout this section, \method is parameterized with $\alpha^* = 1$ at $M^\prime = L$, guaranteeing optimal security against PCA key estimation attacks, as stated in Proposition~\ref{prop:security}.

\paragraph{Image and transforms dataset}
Evaluation is performed over $1000$ $1024 \times 1024$ natural images randomly sampled from the MFlickr dataset, disjoint from the COCO training set. 
To assess capacity under realistic channel degradations, all methods are tested against an extensive set of valuemetric and geometric transforms reported in Appendix~\ref{app:additional_results}(Table~\ref{tab:detailed_results}).

\paragraph{Evaluation protocol}
We evaluate the performance of a watermarking system along three characteristics:
\begin{itemize}
    \item \textit{Capacity:} The watermarking literature usually reports message bit-accuracy without error correcting code. Such a measure makes no sense: even a watermarking system which boasts a bit-accuracy of $0.997$ such as \textsc{PixelSeal} has a probability of decoding error of $1-0.997^{256}=0.53$; far too high to be of any use. We argue that a practical system must use error-correcting codes (ECC). In order to stay agnostic to the specific choice of ECC, as well as to the choice of probability of decoding error, we report the total Shannon capacity (in bits), computed from the empirical codeword bit-accuracy. Note that this empirical evaluation assumes each codeword bit to be independent, which is ensured with the whitening operation in Appendix~\ref{app:whitening}.
    \item \textit{Quality:} For a fair evaluation, we calibrate the watermark power to a fixed $48 \text{dB}$ PSNR across all methods. We assess the perceptual quality of the watermark via the standard LPIPS metric~\cite{zhang2018perceptual}.
    \item \textit{Security:} Evaluated through the $\eta$-security of Definition~\ref{def:security} when Eve uses a PCA attack for key estimation.
\end{itemize}

\subsection{Results}

\paragraph{Capacity -- Figure~\ref{fig:capacity}-\ref{fig:capacity_attacks}}

\method consistently achieves higher empirical capacity than existing baselines under a binary alphabet, achieving around $2.5\times$ the capacity of \textsc{PixelSeal} for both classical valuemetric and geometric operations as well as more recent blind erasure attacks based on diffusion models. Additional results covering extended transforms set and watermark erasure attacks are provided in Appendix~\ref{app:additional_results}(Tables~\ref{tab:detailed_results} and~\ref{tab:erasing_attacks}). 

\begin{figure}[t]
    \centering
    \begin{minipage}[c]{0.32\textwidth}
        \centering
        \includegraphics[width=\linewidth]{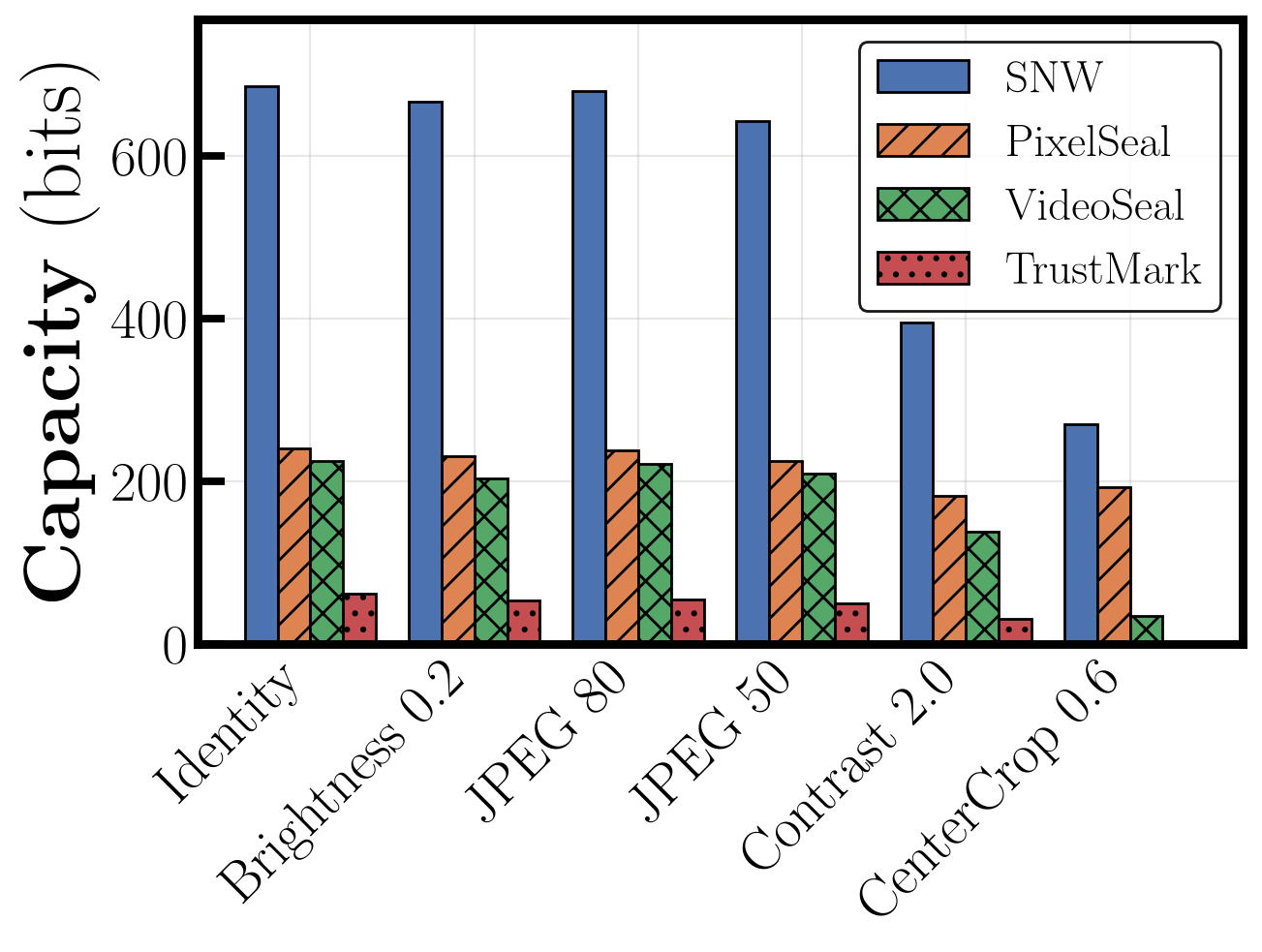}
        \caption{Comparison of recent watermarking systems in terms of capacity in bits over 1000 images. The evaluation is performed against classic image transformations.}
        \label{fig:capacity}
    \end{minipage}
    \hfill 
    \begin{minipage}[c]{0.32\textwidth}
        \centering
        \includegraphics[width=\linewidth]{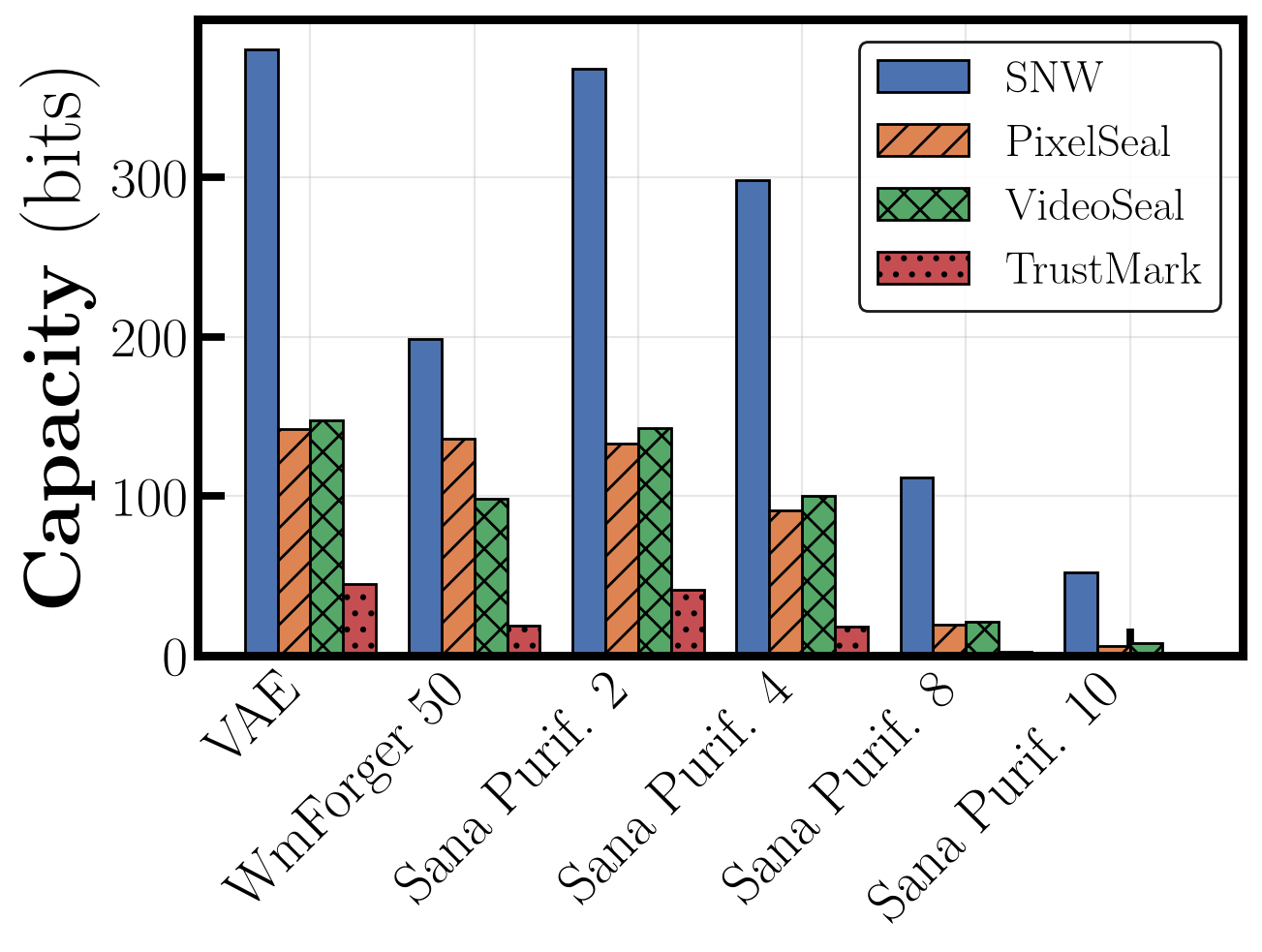}
        \caption{Comparison of recent watermarking systems in terms of capacity in bits over 1000 images. The evaluation is performed against recent watermarking erasure attacks.}
        \label{fig:capacity_attacks}
    \end{minipage}
    \hfill 
    \begin{minipage}[c]{0.32\textwidth}
        \centering
        \resizebox{\linewidth}{!}{%
        \begin{tabular}{l cc}
            \toprule
            \textbf{Method} & LPIPS $(\downarrow)$ & $\eta$-Security $(\uparrow)$ \\
            \midrule
            \textsc{PixelSeal} & 0.0034 & 1.00 \\
            \textsc{VideoSeal} & 0.0025 & 1.00 \\
            \textsc{TrustMark} & \textbf{0.0010} & 1.01 \\
            \method       & 0.0047 & $+\infty$ \\
            \bottomrule
        \end{tabular}%
        }
        \captionof{table}{Comparison between watermarking systems in terms of perceptual quality and security at a fixed watermarking power of 48 dB PSNR. The perceptual quality is measured with the LPIPS.}
        \label{tab:quality_security}
    \end{minipage}
\end{figure}

\paragraph{Quality -- Table~\ref{tab:quality_security}}
Although TrustMark achieves slightly lower perceptual distortion (at the cost of far lower capacity than all other schemes), all evaluated methods remain within the same order of magnitude.

\paragraph{Security -- Table~\ref{tab:quality_security} }
Following the discussion in Section~\ref{sec:channel}, a charitable analysis associates a security ratio $\eta =L(L+1)^{-1} \approx 1$ for keyless systems. Based on Proposition~\ref{prop:security}, under the regime $\alpha^* = 1$ for $M^\prime = L$, \method is perfectly secure against PCA secret key estimation.
\section{Conclusion}
The main impetus of this paper was to pull the focus away from training and architecture in modern watermarking design. Instead, we argued for a careful examination of the implicit channel of these systems. In addition to a lack of security, our analysis revealed an inefficient use of the watermarking space, yielding suboptimal capacity for a given neural network architecture. 
We introduced SNW, a novel post-hoc watermarking system that achieves $2.5 \times$ higher capacity while offering theoretical security guarantees through the use of a secret key. Beyond the question of pure performance, we made the design general enough to accommodate more powerful codebook constructions. In particular, the use of a continuous alphabet and corresponding capacity-achieving codes would be a natural extension, allowing to reach far better capacities at almost no cost.

\bibliography{iclr2027_conference}
\bibliographystyle{iclr2027_conference}

\appendix

\section{Notations}\label{app:notations}

\paragraph{Spaces}
\begin{itemize}
    \item \textbf{Pixel space}: in $\mathbb{R}^D$ with cover observations denoted $\bx$ and watermarked as $\bxwm$.
    \item \textbf{Latent space}: in $\mathbb{R}^L$ during decoding with observations denoted $\bz$ and steering vector denoted as $\bv$. In $\mathbb{R}^{L_e}$ during encoding.
    \item \textbf{Soft-codeword space}: in $\mathbb{R}^{M^\prime}$ with observations denoted as $\tilde{\bc}$.
\end{itemize}

\paragraph{Distributions} We use the standard notation for the standard Gaussian p.d.f $\phi$ and its c.d.f $\Phi$. Other distributions are usually referred to with calligraphic letters $(\mathcal{N})$. 

\paragraph{Other common notation} 
\begin{itemize}
    \item  $\mathbb{S}^{L-1}$: hypersphere embedded in a $L$-dimensionnal real space
    \item $\mathcal{A}^{\Mprime}$: either the binary alphabet $\{0,1\}^{\Mprime}$ or the continuous alphabet $\mathbb{R}^{\Mprime}$. 
    \item $\mathcal{F}$ always refers to the cover distribution -- see Def~\ref{def:cover_distrib}. 
    \item \textbf{Vectors}: vectors $\bv$ use lowercase boldface letters.
    \item \textbf{Functions}: functions $f$ use lowercase regular letters.
    \item \textbf{Constant and random variables}: both constant and r.v. use uppercase, regular Latin letters such as $C$ and $X$. Constants take letters from the start of the alphabet, whereas random variables take letters from the end. Random variables following multivariate distributions are not bolded.
\end{itemize}

\paragraph{Post hoc watermarking system} We refer to a \method system as an error-corrected watermarking system, where:
\begin{itemize}
    \item The set of secret keys $\mathcal{K}$, is given by the set of $L\times M^{\prime}$ matrices with columns summing to $1$. An element of this set is denoted as $\mathbf{U}$.
    \item The decoding projection function $f_d$ is any function from $\mathbb{R}^D$ to $\mathbb{R}^L$.
    \item The encoding projection function $(f_e,f_e^\dagger)$ are any function from $\mathbb{R}^D$ to $\mathbb{R}^{L_e}$ and from $\mathbb{R}^{L_e}$ to $\mathbb{R}^{L_e}$. Note that $L_e$ need not agree with $L$.
\end{itemize}
We assume the redundancy mechanism works at Shannon's capacity. A message is denoted as $\bm \in \{0,1\}^M$ and its representative codeword as $\bc \in \mathcal{A}^{\Mprime}$. 

For convenience, we use a slightly modified version of the sign function defined as:
\begin{align}
    \mathrm{sign}(x) = \begin{cases}
    1 &\text{ if } x > 0\\
    -1 &\text{ else }
    \end{cases}
\end{align}
Importantly, note that $\mathrm{sign}(0) = -1$.

\section{Formal Definitions}\label{app:definitions}
We herein make precise the distinction between a codeword and a message. This leads to a definition of an error-corrected watermarking system.

\begin{appendixdef}[Redundancy mechanism]
    A redundancy mechanism is composed of an encoder and a decoder $(c_{\mathrm{enc}},c_{\mathrm{dec}})$. The encoder $c_{\mathrm{enc}}$ is a bijection between messages and a subset $\mathcal{C} \in \mathcal{A}^{\Mprime}$ called the codebook. 
    For each codeword $\bc$ in the codebook, the decoder $c_{\mathrm{dec}}$ defines an equivalence class $[\bc] = \{\tilde{\bc} \in \mathcal{A}^{M^{\prime}} : c_{\mathrm{dec}}(\tilde{\bc} ) = \bc \}$ and where $c_{\mathrm{dec}} : \mathcal{A}^{M^{\prime}} \rightarrow  \mathcal{A}^{M^{\prime}}$. We abuse notation and make no distinction between the equivalence class $[\bc]$ and its representative codeword $\bc$.
    \label{def:redundancy_mechanism}
\end{appendixdef}

\begin{appendixdef}[Error-corrected watermarking system]
A $M$-bit watermarking system $\mathcal{W}$ equipped with a redundancy mechanism $(c_{\mathrm{enc}},c_{\mathrm{dec}})$ replaces its embedding functions $(e_k)_{k\in\keyset}$ by $e_k^{(c)} \triangleq e_k \circ c_{\mathrm{enc}}$. Its decision mechanism is replaced by a function $d^{(c)}_k \triangleq  c_{\mathrm{enc}}^{-1} \circ c_{\mathrm{dec}} \circ d_k$ where $d_k: \mathbb{R}^L \rightarrow \mathcal{A}^{M^{\prime}}$. We call $d^{(c)}_k$ the error-corrected decision mechanism; we still call $d_k$ the decision mechanism.
\label{def:error_corrected_ws}
\end{appendixdef}

We also characterize implicitly the fact that a watermarking system should not favor a message over another when it receives a cover image with the notion of cover distribution:
\begin{appendixdef}[Cover distribution with binary alphabet]
    Let $\mathcal{F}$  be a probability distribution with support in $\mathbb{R}^D$. It is said to be a cover distribution for a $M$-bit watermarking system $\mathcal{W}$ with a binary alphabet (i.e $\mathcal{A}^{\Mprime} \equiv \{0,1\}^{\Mprime}$) iff, for all secret keys $k \in \keyset$:
    \begin{equation}
        d_k\left(f(X)\right) \sim \mathcal{B}^M\left(\frac{1}{2}\right),
    \end{equation}
    where $X \sim \mathcal{F}$ and $\mathcal{B}^M$ is a $M$-dimensional Bernoulli distribution.
    \label{def:cover_distrib}
\end{appendixdef}

\section{Training details}\label{app:training_details}

In this section, we provide additional details on the training of the projection function of SNW. 
We do not claim that the procedure presented here is the unique or optimal training method to achieve the objectives introduced in Section~\ref{sec:method}.
The model has been trained on the COCO 2017 train set. 

\subsection{Architecture}
We adopt the neural network architecture proposed by PixelSeal. We remove the classification linear head to train only the ConvNeXt projection and the embedder. 
While PixelSeal embeds a binary message of dimension $M^\prime = 256$, we modify the input of the U-Net embedder to accept a continuous random unit vector of dimension $L = 768$. 
Keeping the PixelSeal backbone intact is essential to demonstrate that, with an identical architecture, their end-to-end training pipeline was suboptimal in terms of both capacity and security.

\subsection{Losses}
\label{app:losses}

We detail below the losses used during the training of the encoder-projection model. 
As stated in Section~\ref{sec:method}, the projection aims to:
\begin{itemize}
    \item Accurately embed and extract an arbitrary target unit vector,
    \item Ensure isotropy of the extracted representations across distinct codewords and non-watermarked images. 
\end{itemize}

\paragraph{Alignment losses $\mathcal{L}_{\text{align}}$}

\begin{itemize}
    \item \textbf{$\mathcal{L}_{\text{align}}^{\text{all}}$} Enforces that each spatial vector of the unpooled feature map aligns with the target vector $\bv$:
    \begin{equation}
    \mathcal{L}_{\text{align}}^{\text{all}} = \frac{1}{H^\prime W^\prime} \sum_{i=1}^{H^\prime}\sum_{j=1}^{W^\prime}
        \left(1 - \langle \bar{\bv}_{i,j}, \bv \rangle\right),
    \label{eq:loss_align}
    \end{equation}
    where $\bar{\bv}_{i,j} = \tilde{\bv}_{i,j} / \|\tilde{\bv}_{i,j}\|_2$.
    \item \textbf{$\mathcal{L}_{\text{align}}^{\text{pool}}$} Enforces that the pooled representation aligns with the target vector $\bv$:
    \begin{equation}
        \mathcal{L}_{\text{align}}^{\text{pool}} = 1 - \langle \tilde{\bv}, \bv \rangle.
        \label{eq:loss_align_pool}
    \end{equation}
\end{itemize}

\paragraph{Isotropy losses $\mathcal{L}_{\text{iso}}$}
\begin{itemize}
    \item \textbf{$\mathcal{L}_{\text{iso}}^{\text{same}}$} Penalizes correlations between representations of the same host image $\bx^{(i)}$ encoded with two independent targets $\bv_1$ and $\bv_2$:
    \begin{equation}
    \mathcal{L}_{\text{iso}}^{\text{same}} =
        \frac{1}{B} \sum_{i=1}^{B} \langle \tilde{\bv}_{1}^{(i)}, \tilde{\bv}_{2}^{(i)}\rangle^2.
    \label{eq:loss_iso_same}
    \end{equation}
    \item \textbf{$\mathcal{L}_{\text{iso}}^{\text{cross}}$} Penalizes correlations between representations of different host images $\bx^{(i)}$ and $\bx^{(j)}$ encoded with independent targets $\bv_1$ and $\bv_2$:
    \begin{equation}
    \mathcal{L}_{\text{iso}}^{\text{cross}} =
        \frac{1}{B(B-1)} \sum_{i \neq j} \langle \tilde{\bv}_{1}^{(i)}, \tilde{\bv}_{2}^{(j)} \rangle^2.
    \label{eq:loss_iso_cross}
    \end{equation}
    \item \textbf{$\mathcal{L}_{\text{iso}}^{\text{host}}$} Penalizes correlations between representations of unwatermarked host images $\bx^{(i)}$ and $\bx^{(j)}$:
    \begin{equation}
    \mathcal{L}_{\text{iso}}^{\text{host}} =
        \frac{1}{B(B-1)} \sum_{i \neq j} \langle \tilde{\bv}^{(i)}, \tilde{\bv}^{(j)} \rangle^2.
    \label{eq:loss_iso_host}
    \end{equation}
\end{itemize}

\paragraph{Quality losses $\mathcal{L}_{\text{qual}}$}
\begin{itemize}
    \item \textbf{$\mathcal{L}_{\text{qual}}^{\text{PSNR}}$} Hinge loss on the PSNR bounded by a target budget to ensure watermarking power:
    \begin{equation}
    \mathcal{L}_{\text{qual}}^{\text{PSNR}} =
        \operatorname{ReLU}\!\big(\tau(t) - \mathrm{PSNR}(\bx, \bx_w)\big),
    \label{eq:loss_psnr}
    \end{equation}
    where $\tau(t)$ is the target PSNR at step $t$ and $\bx_w$ is the host image $\bx$ watermarked with codeword $\bc$.
    \item \textbf{$\mathcal{L}_{\text{qual}}^{\text{LPIPS}}$} LPIPS loss to ensure perceptual quality:
    \begin{equation}
    \mathcal{L}_{\text{qual}}^{\text{LPIPS}} = \mathrm{LPIPS}(\bx, \bx_w).
    \label{eq:loss_lpips}
    \end{equation}
\end{itemize}

\paragraph{Anti-spoofing loss $\mathcal{L}_{\text{spoof}}$}

\begin{itemize}
    \item \textbf{$\mathcal{L}_{\text{spoof}}$} Enforces that a watermark residual $\bw_1^{(i)} = \bx_{w_1}^{(i)} - \bx^{(i)}$ yields an orthogonal representation when transferred onto an unrelated host image $\bx^{(j)}$:
    \begin{equation}
    \mathcal{L}_{\text{spoof}} = \frac{1}{B} \sum_{i=1}^{B} \langle f(\bx^{(j)} + \bw_1^{(i)}), \bv_1 \rangle^2,
    \label{eq:antispoofloss}
    \end{equation}
    where $j = (i \bmod B) + 1$, and $\bv_1$ is the target vector embedded into $\bx^{(i)}$. 
    See Appendix~\ref{app:spoof_loss} for further details on the residual spoofing attack.
\end{itemize}

\subsection{Training Procedure}

The training pipeline is divided into three stages involving different combinations of losses, as summarized in Table~\ref{tab:losses}.

\begin{table}[h]
    \centering
    \begin{tabular}{lccc}
    \toprule
        Losses & Stage 1 & Stage 2 & Stage 3 \\
        \midrule
        $\mathcal{L}_{\text{align}}^{\text{all}}$   & \cmark & \xmark & \xmark \\
        $\mathcal{L}_{\text{align}}^{\text{pool}}$  & \xmark & \cmark & \cmark \\
        $\mathcal{L}_{\text{iso}}^{\text{same}}$    & \cmark & \cmark & \cmark \\
        $\mathcal{L}_{\text{iso}}^{\text{cross}}$   & \xmark & \cmark & \cmark \\
        $\mathcal{L}_{\text{iso}}^{\text{host}}$    & \xmark & \cmark & \cmark \\
        $\mathcal{L}_{\text{qual}}^{\text{PSNR}}$   & \cmark & \cmark & \cmark \\
        $\mathcal{L}_{\text{qual}}^{\text{LPIPS}}$  & \xmark & \cmark & \cmark \\
        $\mathcal{L}_{\text{spoof}}$                & \xmark & \xmark & \cmark \\
        \bottomrule
    \end{tabular}
    \caption{Active losses per training stage.}
    \label{tab:losses}
\end{table}

\paragraph{Stage 1: Find the signal} The first stage of training aims to train the model to encode and extract a watermark signal when there are no transforms, while being isotropic. To find this signal and the isotropy, the model needs to remove the host content to keep only the signal. This is why, during the first step, the loss function is:
\begin{equation}
    \mathcal{L}_{1}  = \lambda_{\text{align}}^{\text{all}} \mathcal{L}_{\text{align}}^{\text{all}} + \lambda_{\text{iso}}^{\text{same}}\mathcal{L}_{\text{iso}}^{\text{same}} + \lambda_{\text{qual}}^{\text{PSNR}}\mathcal{L}_{\text{qual}}^{\text{PSNR}},
\end{equation}
with a target PSNR scaled from $ 0$ dB to $ 42$ dB. 
$\mathcal{L}_{\text{align}}^{\text{all}}$, $\mathcal{L}_{\text{iso}}^{\text{same}}$, and $\mathcal{L}_{\text{qual}}^{\text{PSNR}}$ are respectively described by Equations~\ref{eq:loss_align},~\ref{eq:loss_iso_same}, and~\ref{eq:loss_psnr}. 
We specifically use the align loss on all components of the representation before the mean pooling to enforce the model to find a signal. 
To provide robustness, we progressively add transforms to the augmentation layer. 
We use the same set of transforms used by PixelSeal~\cite{soucek_pixel_2025}.

\paragraph{Stage 2: Mean pooling of the signal} The second stage of the training consists of transferring the signal from the components of the representation to the pooled representation while improving the perceptual quality. With this step, we also want to improve the isotropy. 
To do so, we use the following loss function:
\begin{equation}
    \mathcal{L}_{2}  = \lambda_{\text{align}}^{\text{pool}} \mathcal{L}_{\text{align}}^{\text{pool}} + \lambda_{\text{iso}}^{\text{same}}\mathcal{L}_{\text{iso}}^{\text{same}} +
    \lambda_{\text{iso}}^{\text{cross}}\mathcal{L}_{\text{iso}}^{\text{cross}} +
    \lambda_{\text{iso}}^{\text{host}}\mathcal{L}_{\text{iso}}^{\text{host}} +
    \lambda_{\text{qual}}^{\text{PSNR}}\mathcal{L}_{\text{qual}}^{\text{PSNR}} +
    \lambda_{\text{qual}}^{\text{LPIPS}}\mathcal{L}_{\text{qual}}^{\text{LPIPS}}.
\end{equation}

Concerning the alignment, $\mathcal{L}_{\text{align}}^{\text{pool}}$ is used to train the model such that the pooled representation obtained is aligned with the target vector $\bv$. 
Regarding the perceptual quality, we add the LPIPS loss presented in~\ref{eq:loss_lpips}, as well as a JND attenuation at the encoding stage as proposed by PixelSeal~\cite{soucek_pixel_2025}. 
We also continue to enforce isotropy using the three losses $\mathcal{L}_{\text{iso}}^{\text{same}}$, $\mathcal{L}_{\text{iso}}^{\text{cross}}$, and $\mathcal{L}_{\text{iso}}^{\text{host}}$ that respectively request isotropy between the same image watermarked with different codewords,  between different images watermarked with different codewords, and between different non-watermarked images. 

\paragraph{Stage 3: Anti-spoofing fine-tuning}
To address the residual transfer attack, we add a fine-tuning stage with a few steps to train the model such that the detectability of a residual depends on the content of the image.
The loss used is:
\begin{equation}
    \mathcal{L}_3 = \mathcal{L}_2 + \lambda_{\text{spoof}}\mathcal{L}_{\text{spoof}}.
\end{equation}

More details about the anti-spoofing loss function and the residual transfer attack are provided in the Appendix~\ref{app:spoof_loss}.

\section{Proofs}

\subsection{Proposition~\ref{prop:bit_acc}}
\label{subsec:proof_bit_acc}

\begin{proposition}[Binary \method capacity]
    Under the perturbation model defined in Equation~\ref{eq:perturbation_model}, the probability $p(\rho)$ of correctly decoding a bit is given by:
    \begin{equation}
        p(\rho) = \Phi\left(\alpha\frac{\rho}{\sqrt{1 - \rho^2}}\sqrt{\frac{L}{M^\prime}} \right),
    \end{equation}
    where $\Phi$ is the standard Gaussian c.d.f.
    The resulting Shannon capacity is: $C_\rho = M^\prime \left(1 - h_2(p(\rho))\right)$.
\end{proposition}

\begin{proof}

Let $\bv \in \mathbb{S}^{L-1}$ be the target unit direction encoded in the host content $\bx$ for codeword $\bc \in \{-1, 1\}^{M'}$ under the secret key matrix $\bU \in \mathbb{R}^{L \times M'}$. 
As defined in Equation~\ref{eq:steering_vector}, $\bv$ decomposes into the watermark subspace and its orthogonal complement:
\begin{equation}
    \bv = \alpha \, \frac{\bU \bc}{\sqrt{M'}} + \sqrt{1 - \alpha^2} \, \frac{(\mathbf{I}_L - \bU \bU^\top)\bq}{\|(\mathbf{I}_L - \bU \bU^\top)\bq\|_2},
\end{equation}
where $\bU^\top \bU = \mathbf{I}_{M'}$, $\|\bc\|_2 = \sqrt{M'}$, and $\bq \sim \mathcal{N}(\mathbf{0}, \mathbf{I}_L)$. 
By orthogonality, $\bU^\top (\mathbf{I}_L - \bU \bU^\top) = \mathbf{0}$, which yields:
\begin{equation}
    \bU^\top \bv = \alpha \, \frac{\bc}{\sqrt{M'}}.
    \label{eq:proof_u_v}
\end{equation}

We model the extracted normalized latent representation $\tilde{\bv} \in \mathbb{S}^{L-1}$ as a perturbed $\bv$ by a isotropic orthogonal noise component:
\begin{equation}
    \tilde{\bv} = \rho \bv + \sqrt{1 - \rho^2}\bn,
    \label{eq:noise_v}
\end{equation}
where $\bn \sim \mathcal{U}(\mathbb{S}^{L-2})$ is uniformly distributed on the unit sphere of the hyperplane orthogonal to $\bv$ ($\|\bn\|_2 = 1$, $\langle \bn, \bv \rangle = 0$).

Projecting $\tilde{\bv}$ onto the secret projection matrix $\bU$ gives:
\begin{align}
    \bU^\top \tilde{\bv} &= \rho \bU^\top \bv + \sqrt{1 - \rho^2} \bU^\top \bn \\
    &= \rho \alpha \frac{\bc}{\sqrt{M^\prime}} + \sqrt{1 - \rho^2}\bU^\top \bn 
\end{align}

Asymptotically, we have:
\begin{equation}
    (\bU^\top \bn)_j \sim \mathcal{N}\left(0, \frac{1}{L}\right)
\end{equation}
with
\begin{align}
    \mathbb{E}[(\bU^\top \bn)_j] &= 0 \\
    \operatorname{Var}[(\bU^\top \bn)_j] &= \frac{1}{L}.
\end{align}

We examine the decision on the $j$-th bit. Let $S_j = \alpha \frac{\rho}{\sqrt{M^\prime}} c_j$ denote the projected signal component and $B_j = \sqrt{1 - \rho^2} (\bU^\top \bn)_j$ denote the projected noise component. 
The noise term follows:
\begin{equation}
    B_j \sim \mathcal{N}\left(0, \, \sigma_B^2\right), \quad \text{with } \sigma_B = \sqrt{\frac{1 - \rho^2}{L}}.
\end{equation}

Without loss of generality, let's consider $c_j = +1$.
The probability of success is given by:

\begin{equation}
    \mathbb{P}\left(S_j + B_j > 0 \right) = \mathbb{P}\left(\alpha\frac{\rho}{\sqrt{M^\prime}} + B_j > 0 \right)
\end{equation}

Consequently, applying the standard normal CDF:
\begin{equation}
    p(\rho) = \Phi\left(\alpha \frac{\rho}{\sqrt{1-\rho^2}} \sqrt{\frac{L}{M^\prime}} \right).
\end{equation}

\end{proof}

\subsection{Proposition~\ref{prop:security}}
\label{subsec:proof_security}

\begin{proposition}[\method PCA security]
    A \method system is $\eta L$-secure against PCA attack, requiring at least $N = \eta L$ watermark observations to retrieve the watermark subspace, with:
\begin{equation}
    \frac{N}{L} = \eta = \max\left(1 , \frac{\left( 1 - \sqrt{\lambda_1 M'} \right)^2}{L \left( \sqrt{\lambda_1} - \frac{1}{\sqrt{L}} \right)^2}\right).
\end{equation}
In particular, by setting $\alpha$ to $\alpha^* := \sqrt{\frac{M^\prime}{L}}$ we have that $N \rightarrow \infty$ and the system is said to be \textit{perfectly secure} against PCA attacks.
\end{proposition}

\begin{proof}

Let $\bv \in \mathbb{S}^{L-1}$ be the directional embedding vector defined in Equation~\ref{eq:steering_vector}:
\begin{equation}
    \bv = \alpha \, \frac{\bU \bc}{\sqrt{M^\prime}} + \sqrt{1 - \alpha^2} \, \bq^\prime, \quad \text{with} \quad \bq^\prime = \frac{(\mathbf{I}_L - \bU \bU^\top)\bq}{\|(\mathbf{I}_L - \bU \bU^\top)\bq\|_2},
\end{equation}
where $\bU \in \mathbb{R}^{L \times M'}$ satisfies $\bU^\top \bU = \mathbf{I}_{M'}$, $\bc \in \{-1, 1\}^{M'}$, and $\bq \sim \mathcal{N}(\mathbf{0}, \mathbf{I}_L)$.

We can compute the following covariance matrix:
\begin{align}
    &\mathbb{E}[\bq^\prime {\bq^\prime}^\top] = \frac{1}{L - M^\prime} (\mathbf{I}_L - \bU \bU^\top), \\
    &\mathbb{E}[\bc \bc^\top] = I_{M^\prime}, \\
    &\mathbb{E}[\bv \bv^\top] = \frac{\alpha^2}{M^\prime} \bU \bU^\top + \frac{1 - \alpha^2}{L - M^\prime} (\mathbf{I}_L - \bU \bU^\top).
\end{align}

Under the latent perturbation model $\tilde{\bv} = \rho \bv + \sqrt{1 - \rho^2} \, \bn$ with isotropic noise $\bn \sim \mathcal{U}(\mathbb{S}^{L-2})$ independent of $\bv$ (where $\mathbb{E}[\bn \bn^\top] = \frac{1}{L}\mathbf{I}_L$), the population covariance matrix of the watermarked latents is:
\begin{align}
    \mathbf{\Sigma} = \mathbb{E}[\tilde{\bv} {\tilde{\bv}}^\top] &= \rho^2 \mathbb{E}[\bv \bv^\top] + (1 - \rho^2) \mathbb{E}[\bn \bn^\top]  \\
    &= \rho^2 \frac{\alpha^2}{M'} \bU \bU^\top + \rho^2 \frac{1 - \alpha^2}{L - M'} (\mathbf{I}_L - \bU \bU^\top) + \frac{1 - \rho^2}{L} \mathbf{I}_L.
    \label{eq:proof_cov}
\end{align}

Because $\bU \bU^\top$ and $(\mathbf{I}_L - \bU \bU^\top)$ define mutually orthogonal projection operators, $\mathbf{\Sigma}$ is diagonalized in the basis of $\bU$ and its orthogonal complement.
It exhibits two distinct population eigenvalues:
\begin{align}
    \lambda_1 &= \rho^2 \frac{\alpha^2}{M^\prime} + \frac{1 - \rho^2}{L} \quad (\text{ for } M^\prime \text{ dimensions}),  \\
    \lambda_2 &= \rho^2 \frac{1 - \alpha^2}{L - M'} + \frac{1 - \rho^2}{L} \quad (\text{ for } L - M^\prime \text{ dimensions}). 
\end{align}

According to the Marchenko-Pastur distribution of eigenvalues for random matrix theory, the support of the non-watermarked latent space is:
\begin{equation}
    \mathcal{S}_0 = \left[ \lambda_0 \left(1 - \sqrt{\frac{L}{N}} \right)^2, \lambda_0 \left(1 + \sqrt{\frac{L}{N}} \right)^2 \right]
\end{equation}
where $N$ is the number of observations and $\lambda_0 = \frac{1}{L}$.

The support of the eigenvalues of the covariance matrix from the watermarked latent is:
\begin{equation}
    \mathcal{S}_1 = \left[ \lambda_1 \left(1 - \sqrt{\frac{M^\prime}{N}} \right)^2, \lambda_1 \left(1 + \sqrt{\frac{M^\prime}{N}} \right)^2 \right]
\end{equation}

We consider a model to be secure against PCA for $N = \eta L$ observations if $\mathcal{S}_1 \subseteq \mathcal{S}_0$.

If $\lambda_1 < \lambda_0$:

\begin{align}
    \lambda_1 \left( 1 - \sqrt{\frac{M^\prime}{N}}\right)^2 &= \lambda_0 \left( 1 - \sqrt{\frac{L}{N}}\right)^2 \\
    \sqrt{\lambda_1} - \sqrt{\lambda_0} &= \frac{\sqrt{\lambda_1 M^\prime} - \sqrt{\lambda_0 L}}{\sqrt{N}} \\
    \frac{N}{L} = \eta &= \frac{\left( \sqrt{\lambda_1 M^\prime} - 1 \right)^2}{L\left( \sqrt{\lambda_1} - \frac{1}{\sqrt{L}} \right)^2}
\end{align}

If $\lambda_1 > \lambda_0$, the result is the same and the proof is analogous.

Finally, since forming a full-rank empirical covariance matrix in $\mathbb{R}^L$ requires at least $N \ge L$ independent observations, we obtain the security threshold:
\begin{equation}
    \eta = \max\left(1 , \frac{\left( 1 - \sqrt{\lambda_1 M'} \right)^2}{L \left( \sqrt{\lambda_1} - \frac{1}{\sqrt{L}} \right)^2}\right).
\end{equation}

\end{proof}

\subsection{Proposition~\ref{prop:security} (Perfect Security)}
\label{subsec:proof_perfect_security}

\begin{proposition}[SNW perfect PCA security]
    The perfect security regime of SNW against PCA estimation is achieved when:
    \begin{equation}
    \alpha^* = \sqrt{\frac{M^\prime}{L}}
\end{equation}
\end{proposition}

\begin{proof}
Consider the population covariance matrix $\mathbf{\Sigma} = \mathbb{E}[\tilde{\bv} {\tilde{\bv}}^\top]$ derived in Equation~\ref{eq:proof_cov}:
\begin{equation}
    \mathbf{\Sigma} = \rho^2 \frac{\alpha^2}{M^\prime} \bU \bU^\top + \rho^2 \frac{1 - \alpha^2}{L - M^\prime} (\mathbf{I}_L - \bU \bU^\top) + \frac{1 - \rho^2}{L} \mathbf{I}_L.
\end{equation}
The population spectrum is characterized by the two eigenvalues $\lambda_1$ on the $M^\prime$-dimensional watermark subspace and $\lambda_2$ on the $(L - M^\prime)$-dimensional complementary subspace:
\begin{equation}
    \lambda_1 = \rho^2 \frac{\alpha^2}{M^\prime} + \frac{1 - \rho^2}{L}, \qquad
    \lambda_2 = \rho^2 \frac{1 - \alpha^2}{L - M^\prime} + \frac{1 - \rho^2}{L}.
\end{equation}

Perfect security against PCA is reached when the population covariance matrix is strictly isotropic, which occurs if $\lambda_1 = \lambda_2$:

\begin{align}
    \rho^2 \frac{(\alpha^*)^2}{M^\prime} + \frac{1 - \rho^2}{L} &= \rho^2 \frac{1 - (\alpha^*)^2}{L - M^\prime} + \frac{1 - \rho^2}{L}  \\
    \iff \quad \alpha^* &= \sqrt{\frac{M^\prime}{L}}.
\end{align}

\end{proof}

\section{Watermarking security and adversarial robustness}\label{app:secu_adv_robustness}
The concept of security in the modern post-hoc literature is nebulous.
It is seldom discussed in papers dedicated to novel architectures. Adversarial robustness is discussed only as a mean to improve image quality (HiDDeN's strategy ~\cite{ferrari_hidden_2018} or decoding performance (\textsc{PixelSeal}'s adversarial training~\cite{soucek_pixel_2025}). \textsc{TrustMark}~\cite{bui_trustmark_2025} mentions neither security nor adversarial robustness in the whole paper. To the best of our knowledge \textsc{SynthID-Image} is the \textit{only} work of this type to mention and discuss watermarking security specifically~\cite{gowal_synthid-image_2025}[Section 6].

Again, the \textsc{SynthID-Image} report is instructive. It conflates genuine watermarking security (watermark forgery, removal, and secret extraction) with adversarial robustness and, more surprisingly, with model extraction vulnerabilities. The confusion is especially interesting, since the authors seem to imply that adversarial machine learning is the main tool for attacking a watermarking system.

Similarly, the literature specialized in the security of watermarking does not differentiate security when talking about \begin{itemize}
\item Signal estimation and forgery based on classical watermarking analysis and signal processing~\cite{tarhini_neural_2026,bas_ai_2025}.
\item Spoofing and erasure attacks using adversarial machine learning techniques~\cite{gesny_modern_2026,soucek_transferable_2025,jiang_evading_2023}.
\end{itemize}

The first category is what would truly be called watermarking security in the classic literature. Citing the reference definition of Kalker~\cite{kalker_considerations_2001}, the goal of these attacks is to obtain:
\begin{quote}
    [...] unauthorized [...] access to the raw watermarking channel.
    In other words, watermark security refers to the inability of unauthorized users to remove, detect and estimate, write, or modify the raw watermarking bits. In particular, watermark security is not concerned with the semantics of the watermarking bits, but solely with the physical presence of the watermarking bits.
\end{quote}
The PCA attack studied in this paper falls into this category: estimating the key grants complete access to the watermarking channel, in the sense that any arbitrary image can now be spoofed with any codeword. This definition of security has a long history, culminating in the refined definitions of \textit{equivocation}~\cite{cayre_watermarking_2005-1} and \textit{effective key length}~\cite{bas_new_2013}, with our definition being a simplification of the combination of the two.

This is in contrast to adversarial machine learning, which can be leveraged against a specific part of the watermarking system -- the decoding projection. When they were not based on DNNs, such attacks were called "oracle attacks" or "sensitivity attacks" in the classical literature -- for example, see Broken-Arrows' "snake traps"~\cite{furon_broken_2008}[Section 5.2], or work from the early 1990's~\cite{cox1997public,linnartz1998analysis}. They do allow to perform erasure and spoofing attacks. But they do not, by themselves, fully compromise the system. Their output is a single perturbation, not always transferable (across images, systems, models, etc$\ldots$), that allows a specific operation (erasure, spoofing, $\ldots$.

\subsection{Adversarial robustness and Lipschitz networks}

Adversarial robustness assesses how much the decoding projection function $f$ against an adversarial pixel perturbation. In our scenario, such perturbations are crafted to spoof a watermark signal: 

\begin{definition}[Adversarial robustness]
     Let $\Psi : \mathbb{R}^D \rightarrow \mathbb{R}^D $ be an attack such that:
     
     \begin{equation}
         d_k(f_d(\Psi(\bx))) = \bc, \forall \bx \in \mathbb{R}^D, \forall k \in \mathcal{K}, \forall \bc \in \mathcal{A}^{M^\prime}
     \end{equation}
     
     We say that a (decoding) projection $f_d$ is $\epsilon$-adversarial-robust against $\Psi$ iff :
      \begin{equation}
          \mathbb{E}\left[||\Psi(X) - X||_2\right] \leq \epsilon 
      \end{equation}
      where the expectation is taken over $(X, C, K)$ triplets, with $X$ sampled from the cover distribution $\mathcal{F}$, and $(C, K)$ sampled uniformly from their respective set.
\end{definition}

In theory, we are not required to evaluate the adversarial robustness of a projection function empirically. We can leverage the fact that DNNs are often considered to be Lipschitz functions.
Formally, if a projection function $f_d$ is $L_{f_d}$-Lipschitz with respect to the $\ell_2$ norm, for any perturbation $\boldsymbol{\epsilon}$:
\begin{equation}
    \|f_d(\bx + \boldsymbol{\epsilon}) - f_d(\bx)\|_2 \le L_{f_d} \|\boldsymbol{\epsilon}\|_2.
\end{equation}

Consequently, inducing a target latent displacement $\|f_d(\bx + \boldsymbol{\epsilon}) - f_d(\bx)\|_2$ required to cross a decision boundary necessitates a pixel-space distortion strictly lower-bounded by:
\begin{equation}
\|\boldsymbol{\epsilon}\|_2 \ge \frac{\|f_d(\bx + \boldsymbol{\epsilon}) - f_d(\bx)\|_2}{L_{f_d}}.
\label{eq:lipschitz_bound}
\end{equation}

When $L_{f_d}$ is large, this theoretical lower bound vanishes, allowing imperceptible pixel noise to displace latents across the decision boundary.

\subsection{Lipschitz estimation}
\label{app:lipschitz_estimation}

Computing the global Lipschitz constant of a deep neural network is computationally intractable. 
We evaluate local stability by estimating the mean local Lipschitz constant across 100 randomly sampled images from the ImageNet dataset~\cite{deng2009imagenet} for each projection function $f$.

The local Lipschitz constant $L_{f_d}(\bx_0)$ of ${f_d}$ at a point $\bx_0$ corresponds to the spectral norm of its Jacobian matrix $J_f(\bx_0)$, which bounds the first-order sensitivity to local perturbations:
\begin{equation}
    L_{f_d}(\bx_0) = \| J_f(\bx_0)\|_2 = \sigma_{\text{max}}\left(J_{f_d}(\bx_0)\right),
\end{equation}
where $\sigma_{\text{max}}$ is the largest singular value. 

We estimate $\sigma_{\text{max}}$ using the power iteration algorithm proposed by~\cite{miyato2018spectral} because materializing the full Jacobian matrix is too computationally expensive for high-dimensional inputs such as images.  

We report estimates of the Lipschitz constant for the watermarking systems studied in this paper in Table~\ref{tab:lipschitz}. For comparison, we also add Broken-Arrows~\cite{furon_broken_2008} as a classical watermarking system that guarantees $L_f = 1$ by design through its Discrete Wavelet Transform (DWT), ensuring Euclidean distance conservation. 
In contrast, deep neural network projection functions are trained without Lipschitz constraints, leading to elevated Lipschitz constants and vulnerability to adversarial watermark erasure~\cite{gesny_modern_2026}.  
\begin{table}[h]
    \centering
    \begin{tabular}{ll|c}
        \toprule
         \textbf{WM} & Projection function & Lipschitz constant $L_f$\\
         \midrule
         \textsc{PixelSeal} & ConvNeXT & 1079 \\
         \textsc{VideoSeal} & ConvNeXT & 595 \\
         \textsc{TrustMark} & ResNet & 120 \\
         Broken-Arrows & DWT & 1\\
         \bottomrule
    \end{tabular}
    \caption{Comparison of empirical estimates of the Lipschitz constant $L_f$ of the projection function across schemes. Classical transforms guarantee distance preservation ($L_f = 1$), whereas unconstrained neural projection functions exhibit large empirical Lipschitz constants, exposing them to low-distortion removal. The estimation is performed over 100 ImageNet images.}
    \label{tab:lipschitz}
\end{table}

\section{Whitening}
\label{app:whitening}

Deep watermarking detectors yield biased and correlated output logits~\cite{gesny2026guidance}[Appendix C.1.1]. 
To ensure a fair comparison across the methods, we follow the whitening procedure established in~\cite{gesny2026guidance}[Appendix C.1.1]. 
For completeness, we summarize the methodology below. 

For each detector $\phi$, we compute the logits vector over $n = 10^6$ natural images from the MFlickr dataset.
We first compute the empirical bias:
\begin{equation}
    \mathbf{b}_\phi = \frac{1}{n} \sum_{i=1}^n\phi\left(\bx^{(i)}\right),
    \label{eq:bias_eq}
\end{equation}
and the covariance matrix:
\begin{equation}
    \mathbf{\Sigma}_\phi = \frac{1}{n-1} \sum_{i=1}^n \left( \phi\left(\bx^{(i)} \right) - \mathbf{b}_\phi\right)\left( \phi\left(\bx^{(i)} \right) - \mathbf{b}_\phi\right)^\top.
    \label{eq:cov_mat}
\end{equation}
Since $\mathbf{\Sigma}_\phi$ is symmetric positive semi-definite, we compute its eigendecomposition:
\begin{equation}
    \mathbf{\Sigma}_\phi = \mathbf{V}_\phi \mathbf{\Lambda}_\phi \mathbf{V}_\phi^\top,
\end{equation}
where $\mathbf{\Lambda}_\phi = \operatorname{diag}(\lambda_1, \dots, \lambda_d)$ contains the eigenvalues and $\mathbf{V}_\phi$ the corresponding orthonormal eigenvectors. 

The projection matrix is then defined as:
\begin{equation}
    \mathbf{W}_\phi = \mathbf{\Lambda}_{\phi}^{-\frac{1}{2}}.
\end{equation}
Finally, the whitened detector output is given by:
\begin{equation}
    \phi_w(\bx) = \mathbf{W}_\phi \left(\phi\left(\bx\right) - \mathbf{b}_\phi\right).
    \label{eq:whitened_decoder}
\end{equation}

\section{Residual transfer attacks}
\label{app:spoof_loss}

A critical requirement for watermarking is the content-dependency: the watermark signal must be tied to the host image, preventing an adversary from spoofing the watermark by transplanting the residual from a watermarked content. 
Let $\bx_i^{(1)} = \bx_i + \boldsymbol{\delta}_i^{(1)}$ be the watermarked version of host image $\bx_i$, where $\boldsymbol{\delta}_i^{(1)}$ is the residual watermark embedded in the pixel space.

While vanilla SNW exhibits residual leakage across hosts, we fix this by fine-tuning the model using the following loss:
    \begin{equation}
    \mathcal{L}_{\text{spoof}} = \frac{1}{B} \sum_{i=1}^{B} \langle f(\bx^{(j)} + \bw_1^{(i)}), \bv_1 \rangle^2,
    \end{equation}
where $j = i + 1 \mod B$, and $\bv_1$ is the target vector of image $\bx^{(i)}$. 

\begin{table}[t]
    \centering
    \small
    \begin{tabular}{l cc}
        \toprule
        \textbf{Method}& \multicolumn{2}{c}{Rate $R_\sigma$ / $R_\sigma\times M^\prime$} \\
        \cmidrule(lr){2-3}
        & Identity & Residual Transfer \\
        \midrule
        PixelSeal               & 0.910 / 232.9 & 0.003 / 0.8 \\
        VideoSeal               & 0.848 / 217.2 & 0.002 / 0.6 \\
        Trustmark               & 0.506 / 50.6 & 0.059 / 5.9 \\
        \midrule
        SNW (w/o anti-spoofing) & 0.883 / 678.0 & 0.113 / 86.5 \\
        SNW (w anti-spoofing)   & 0.888 / 682.2 & 0.000 / 0.2 \\
        \bottomrule
    \end{tabular}
    \caption{Bit accuracy of the methods against residual transfer attacks. The results are computed over 200 MFlickr $1024 \times 1024$ images.}
    \label{tab:residual_transferability}
\end{table}

As shown in Table~\ref{tab:residual_transferability}, SNW is vulnerable to residual transfer attacks.
TrustMark is also vulnerable, whereas PixelSeal and VideoSeal are robust. 
Fine-tuning SNW using the anti-spoofing loss~\ref{eq:antispoofloss} suppresses transferability while preserving decoding capacities.

\section{Additional results}\label{app:additional_results}

\subsection{SNW geometric analysis}\label{app:SNW-geometric}
\begin{figure}[h]
    \centering
    \includegraphics[width=1\linewidth]{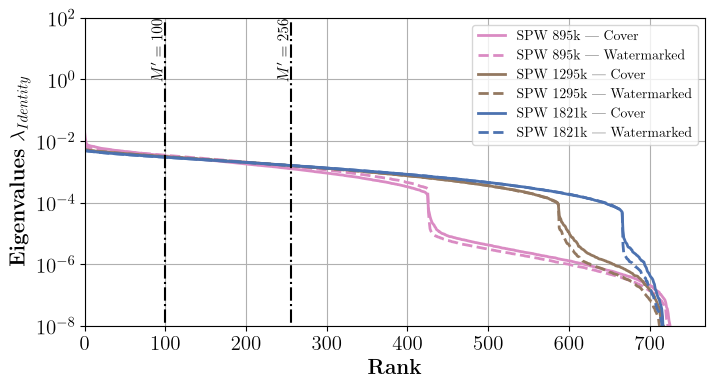}
    \caption{Eigenvalues 
    of the covariance $\Sigma_{\mathrm{Identity}}$ of the latent vector $f_d(\bxwm)$ (i.e. when $t$ is the identity) for different checkpoint of SNW. The number in the legend correspond to the number of training step.}
    \label{fig:latent-analysis-snw}
\end{figure}

We replicate the eigenvalue study in Section~\ref{sec:channel} for different checkpoints of \method in Figure~\ref{fig:latent-analysis-snw}. One can observe our training pipeline design to have succeeded for two reasons:
\begin{enumerate}
    \item As the number of steps increase, we "fill up" the number of robust components, i.e the non-zero eigenvalues. 
    \item The eigenvalues between cover and watermarked content cannot be distinguished. This validates our security analysis: by setting $\alpha=1$ for SNW, one cannot recover the secret key through any PCA or other second-order attack: indeed one cannot distinguish between the covariance of a cover and of a watermarked image.
\end{enumerate}

\subsection{Recent erasing attacks}

In this section, we compare the capacity of modern post-hoc watermarking schemes against recent erasing attacks.
The erasing attacks studied are:
\begin{itemize}
    \item \textbf{VAE Purification}~\cite{nie2022DiffPure} Encode and decode the watermarked image with a frozen VAE. 
    \item \textbf{DiffPure}~\cite{nie2022DiffPure} Invert the last $t$ steps of diffusion and regenerate them. 
    \item \textbf{WM Forger}~\cite{soucek_transferable_2025} Use a preference model to perform a gradient ascent on the image. 
\end{itemize}

\begin{table}[h]
    \centering
    \resizebox{\columnwidth}{!}{%
    \begin{tabular}{l ccccc}
        \toprule
        \textbf{Method} & Identity & VAE & Sana (2 steps) & Sana (4 steps) & Sana (8 steps) \\
        \midrule
        PixelSeal & 0.939 / 240.5 & 0.555 / 142.0 & 0.519 / 132.8 & 0.357 / 91.3 & 0.076 / 19.4 \\
        VideoSeal & 0.884 / 226.2 & 0.576 / 147.6 & 0.557 / 142.6 & 0.392 / 100.4 & 0.084 / 21.5 \\
        TrustMark & 0.632 / 63.2 & 0.453 / 45.3 & 0.415 / 41.5 & 0.184 / 18.4 & 0.023 / 2.3 \\
        SNW       & 0.892 / \textbf{684.7} & 0.495 / \textbf{380.3} & 0.479 / \textbf{368.1} & 0.388 / \textbf{298.3} & 0.145 / \textbf{111.5} \\
        \midrule
        \textbf{Method} & Sana (10 steps) & Sana (20 steps) & WM Forger 50 steps & WM Forger 100 steps &   \\
        \midrule
        PixelSeal & 0.025 / 6.4 & 0.003 / 0.8 & 0.533 / 136.3 & 0.280 / 71.8 &  \\
        VideoSeal & 0.030 / 7.8 & 0.003 / 0.9 & 0.385 / 98.5 & 0.191 / 48.8 &  \\
        TrustMark & 0.0152 / 1.5 & 0.007 / 0.7 & 0.187 / 18.7 & 0.078 / 7.8 &  \\
        SNW       & 0.068 / \textbf{52.2} & 0.001 / 0.8 & 0.312 / \textbf{239.5} & 0.139 / \textbf{106.5} &  \\
        \bottomrule
    \end{tabular}%
    }
    \caption{Capacity of the watermarking systems against recent watermarking erasure attacks. The results are computed over 200 MFlickr $1024 \times 1024$ images.}
    \label{tab:erasing_attacks}
\end{table}

\subsection{Detailed results}\label{app:detailed_results}

In this section, we provide the capacity of modern post-hoc watermarking methods against a full benchmark of classic transformations. 

\begin{table}[h]
    \centering
    \resizebox{\columnwidth}{!}{%
    \begin{tabular}{l ccccc}
        \toprule
        \multirow{2}{*}{\textbf{Method}} & \multicolumn{5}{c}{Rate $R_\sigma$ / $R_\sigma\times M^\prime$} \\
        \cmidrule(lr){2-6}
        & \textbf{Identity} & \textbf{Brightness} $+ 0.2$ & \textbf{Contrast} $\times 2$ & \textbf{JPEG QF} $=80$ & \textbf{JPEG QF} $=50$ \\
        \midrule
        PixelSeal & 0.939 / 240.5 & 0.899 / 230.1 & 0.709 / 181.6 & 0.928 / 237.5 & 0.876 / 224.4 \\
        VideoSeal & 0.878 / 224.7 & 0.793 / 203.1 & 0.539 / 138.1 & 0.863 / 221.0 & 0.816 / 209.0 \\
        TrustMark & 0.611 / 61.1  & 0.530 / 53.0  & 0.309 / 30.9  & 0.545 / 54.5  & 0.491 / 49.1  \\
        SNW       & 0.983 / \textbf{686.0} & 0.868 / \textbf{666.3} & 0.514 / \textbf{394.6} & 0.884 / \textbf{679.0} & 0.837 / \textbf{642.5} \\
        \cmidrule(lr){2-6}
        & \textbf{Gaussian Blur} $3 \times 3$, $\sigma=1$ & \textbf{Rotation} $90^\circ$ & \textbf{Horizontal Flip} & \textbf{Hue} 0.5 & \textbf{Saturation} 1.5 \\
        \midrule
        PixelSeal & 0.939 / 240.3 & 0.863 / 221.0 & 0.940 / 240.6 & 0.736 / 188.5 & 0.928 / 237.6 \\
        VideoSeal & 0.877 / 224.5 & 0.821 / 210.1 & 0.866 / 221.6 & 0.652 / 166.8 & 0.858 / 219.5 \\
        TrustMark & 0.610 / 61.0  & 0.007 / 0.7   & 0.612 / 61.2  & 0.389 / 38.9  & 0.573 / 57.3  \\
        SNW       & 0.892 / \textbf{685.1} & 0.633 / \textbf{486.5} & 0.830 / \textbf{637.7} & 0.706 / \textbf{542.6} & 0.878 / \textbf{674.5} \\
        \cmidrule(lr){2-6}
        & \textbf{Crop} $90\%$ & \textbf{Crop} $80\%$ & \textbf{Crop} $70\%$ & \textbf{Crop} $60\%$ & \textbf{Crop} $50\%$ \\
        \midrule
        PixelSeal & 0.910 / 232.9 & 0.890 / 227.7 & 0.845 / 216.4 & 0.752 / 192.5 & 0.561 / \textbf{143.7} \\
        VideoSeal & 0.768 / 196.7 & 0.636 / 162.9 & 0.389 / 99.6  & 0.131 / 33.6  & 0.019 / 4.8   \\
        TrustMark & 0.694 / 69.4  & 0.538 / 53.8  & 0.016 / 1.6   & 0.008 / 0.8   & 0.007 / 0.7   \\
        SNW       & 0.811 / \textbf{623.2} & 0.725 / \textbf{556.6} & 0.567 / \textbf{435.6} & 0.351 / \textbf{269.6} & 0.133 / 102.0  \\
        \cmidrule(lr){2-6}
        & \textbf{Resize} 0.5 & \textbf{Median Filter} $3 \times 3$ & & & \\
        \midrule
        PixelSeal & 0.937 / 239.9 & 0.938 / 240.1 & & & \\
        VideoSeal & 0.876 / 224.2 & 0.877 / 224.4 & & & \\
        TrustMark & 0.610 / 61.0  & 0.611 / 61.1  & & & \\
        SNW       & 0.891 / \textbf{684.1} & 0.892 / \textbf{685.0} & & & \\
        \bottomrule
    \end{tabular}%
    }
    \caption{Capacity of the watermarking systems against an extensive set of classic image transformations. The results are computed over 1000 MFlickr $1024 \times 1024$ images.}
    \label{tab:detailed_results}
\end{table}

\subsection{Examples}
In this section, we provide a qualitative example of images watermarked with SNW. 

\begin{figure}
    \centering
    \includegraphics[width=\linewidth]{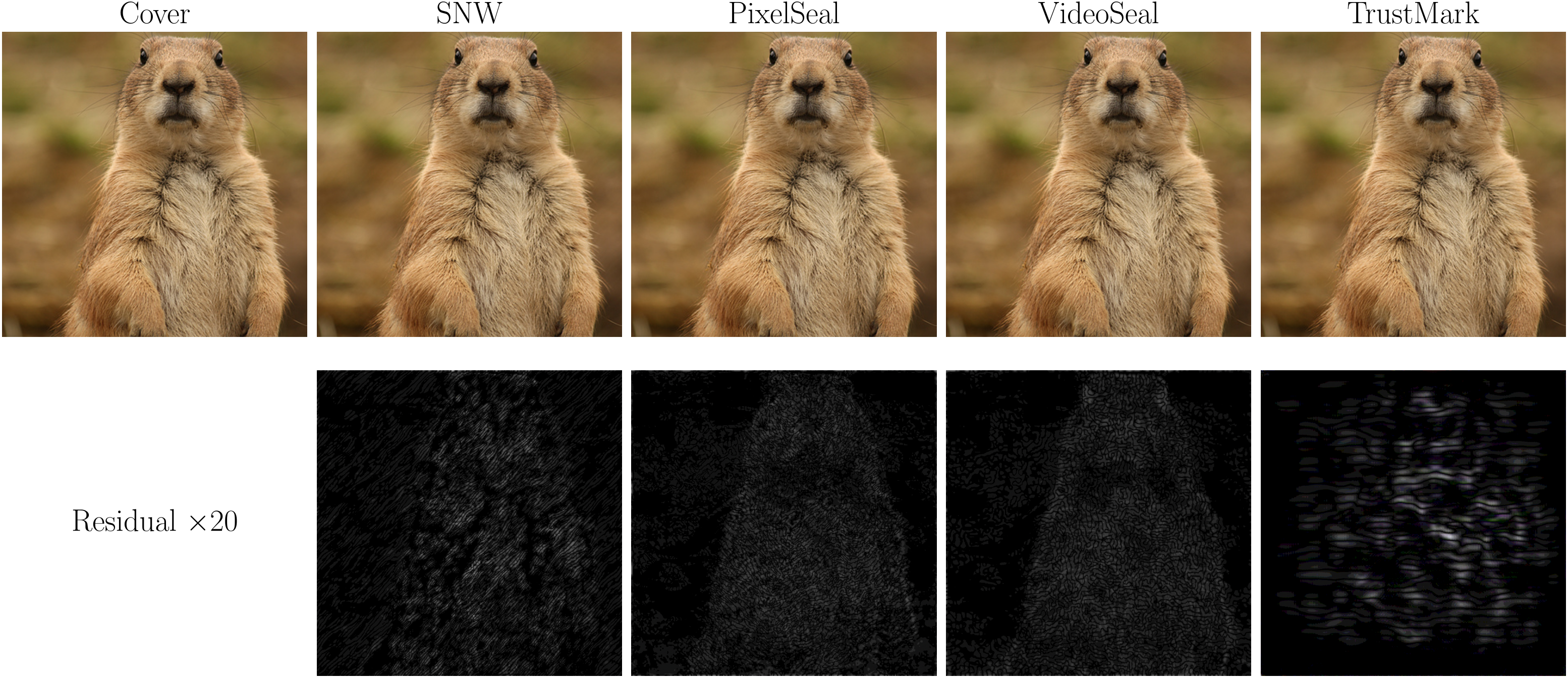}
    \caption{Example of an image watermarked with different methods and associated residuals for a fixed watermark power of 48 dB PSNR.}
    \label{fig:example}
\end{figure}

\end{document}